\documentclass[aps,groupedaddress,superscriptaddress,twocolumn,pra]{revtex4-1}
\usepackage{amsmath}
\usepackage{amssymb}
\usepackage{graphicx}
\usepackage{amsthm}
\usepackage{hyperref}
\usepackage{bbold}
\usepackage{xcolor}
\usepackage{longtable}
\usepackage{booktabs}
\usepackage{subfigure}
\usepackage{adjustbox}
\usepackage{array}
\usepackage{algorithm}
\usepackage{algorithmic}

\begin{document}

\newcommand*{\cl}[1]{{\mathcal{#1}}}
\newcommand*{\bb}[1]{{\mathbb{#1}}}
\newcommand{\ket}[1]{|#1\rangle}
\newcommand{\bra}[1]{\langle#1|}
\newcommand{\inn}[2]{\langle#1|#2\rangle}
\newcommand{\proj}[2]{| #1 \rangle\!\langle #2 |}
\newcommand*{\tn}[1]{{\textnormal{#1}}}
\newcommand*{\1}{{\mathbb{1}}}
\newcommand{\T}{\mbox{$\textnormal{Tr}$}}
\newcommand{\todo}[1]{\textcolor[rgb]{0.99,0.1,0.3}{#1}}
\newcommand{\norm}[1]{\left\lVert#1\right\rVert}
\newcommand{\RN}[1]{\textup{\uppercase\expandafter{\romannumeral#1}}}

\newcommand\sixvonone{\adjustbox{valign=m, vspace=0pt, margin=1ex}{\includegraphics[width=.310\linewidth]{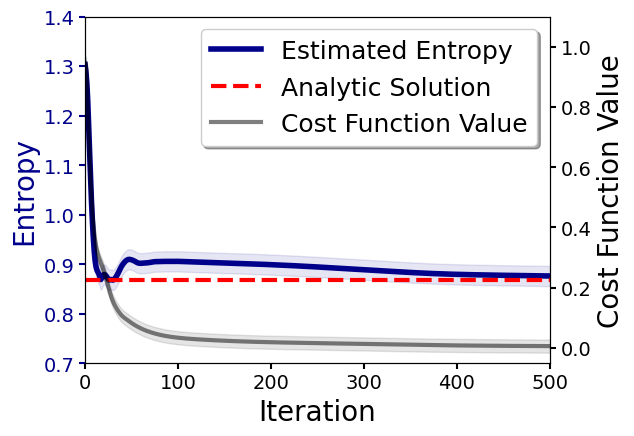}}}
\newcommand\sixrenyione{\adjustbox{valign=m, vspace=0pt, margin=1ex}{\includegraphics[width=.310\linewidth]{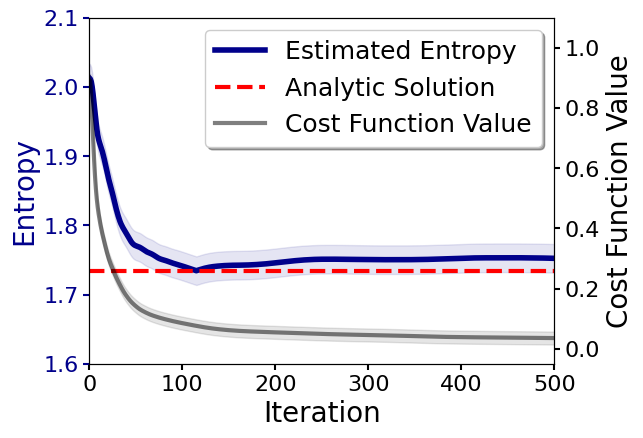}}}
\newcommand\sixtsallisone{\adjustbox{valign=m, vspace=0pt, margin=1ex}{\includegraphics[width=.310\linewidth]{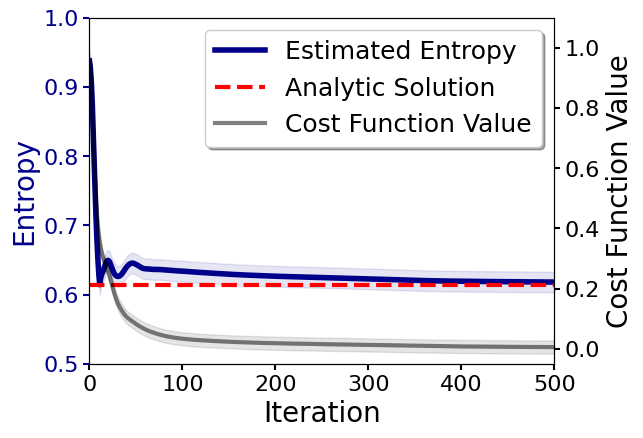}}}
\newcommand\sixtraceone{\adjustbox{valign=m, vspace=0pt, margin=1ex}{\includegraphics[width=.310\linewidth]{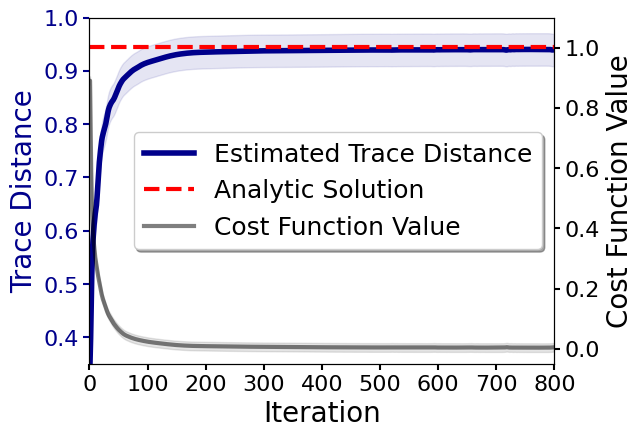}}}
\newcommand\sixfidelityone{\adjustbox{valign=m, vspace=0pt, margin=1ex}{\includegraphics[width=.310\linewidth]{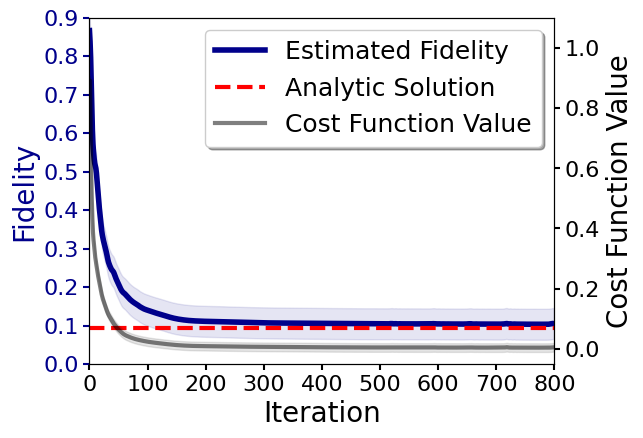}}}

\newcommand\sixvontwo{\adjustbox{valign=m, vspace=0pt, margin=1ex}{\includegraphics[width=.310\linewidth]{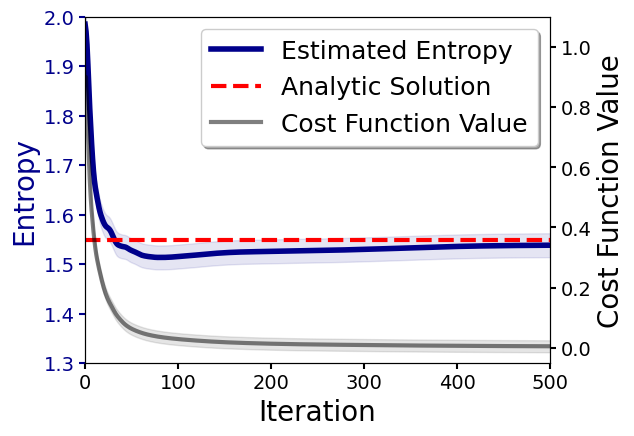}}}
\newcommand\sixrenyitwo{\adjustbox{valign=m, vspace=0pt, margin=1ex}{\includegraphics[width=.310\linewidth]{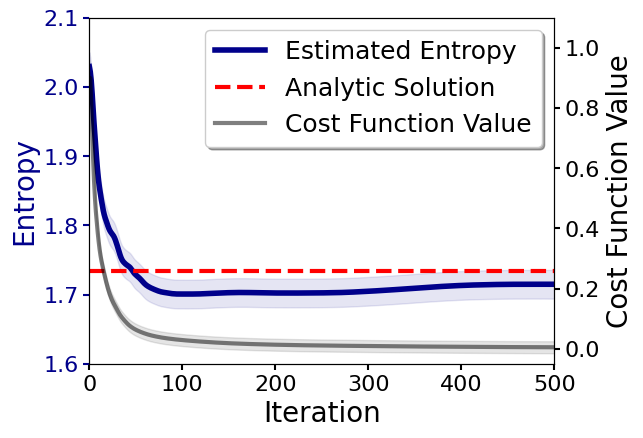}}}
\newcommand\sixtsallistwo{\adjustbox{valign=m, vspace=0pt, margin=1ex}{\includegraphics[width=.310\linewidth]{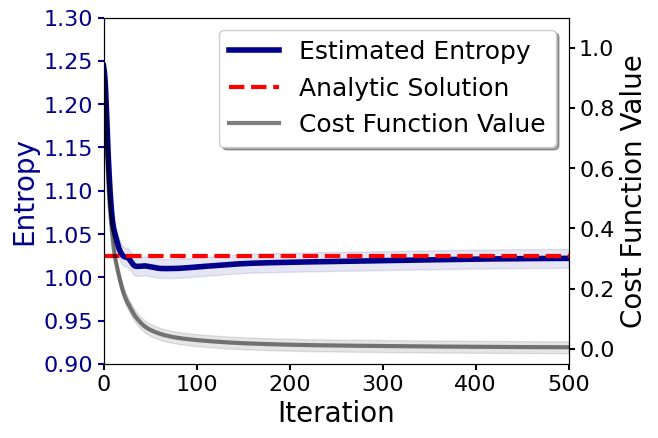}}}
\newcommand\sixtracetwo{\adjustbox{valign=m, vspace=0pt, margin=1ex}{\includegraphics[width=.310\linewidth]{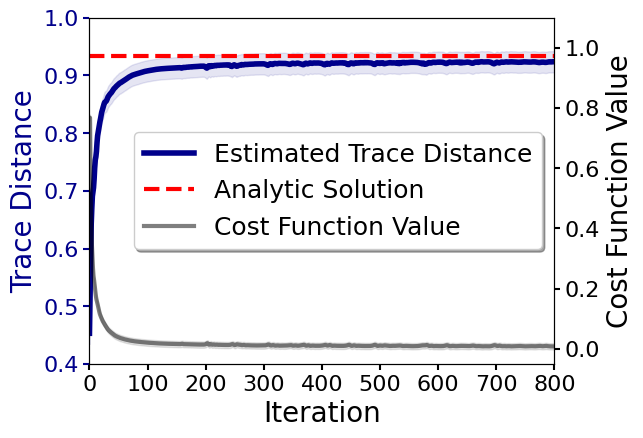}}}
\newcommand\sixfidelitytwo{\adjustbox{valign=m, vspace=0pt, margin=1ex}{\includegraphics[width=.310\linewidth]{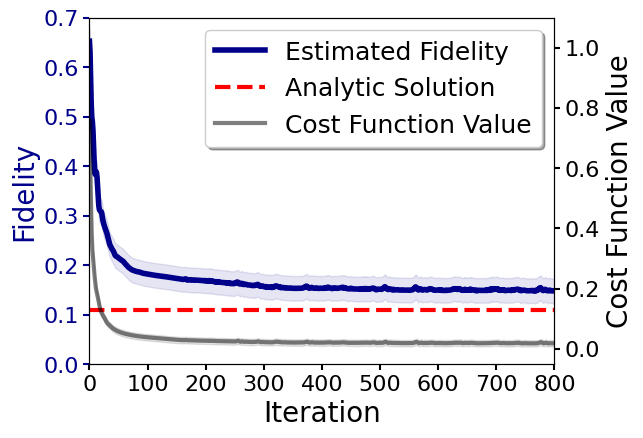}}}

\newcommand\fourvonone{\adjustbox{valign=m, vspace=0pt, margin=1ex}{\includegraphics[width=.310\linewidth]{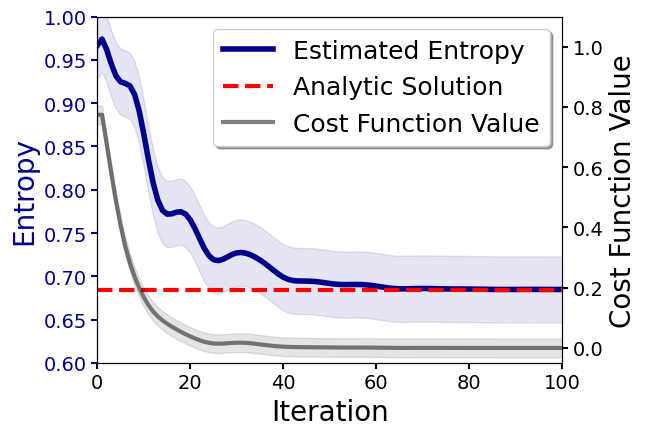}}}
\newcommand\fourrenyione{\adjustbox{valign=m, vspace=0pt, margin=1ex}{\includegraphics[width=.310\linewidth]{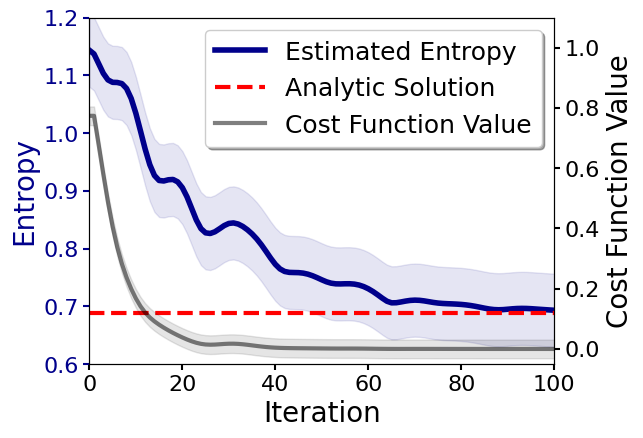}}}
\newcommand\fourtsallisone{\adjustbox{valign=m, vspace=0pt, margin=1ex}{\includegraphics[width=.310\linewidth]{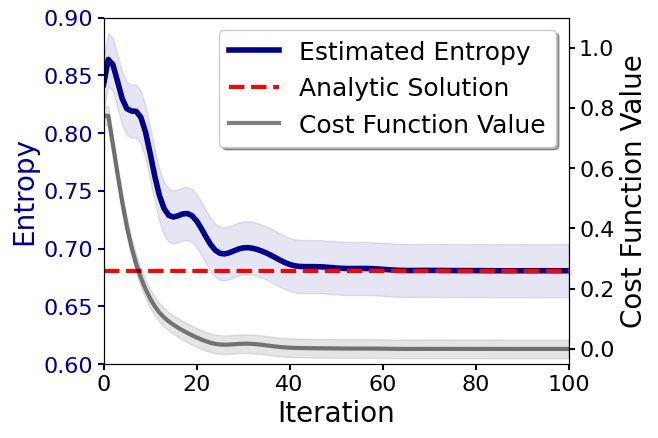}}}
\newcommand\fourtraceone{\adjustbox{valign=m, vspace=0pt, margin=1ex}{\includegraphics[width=.310\linewidth]{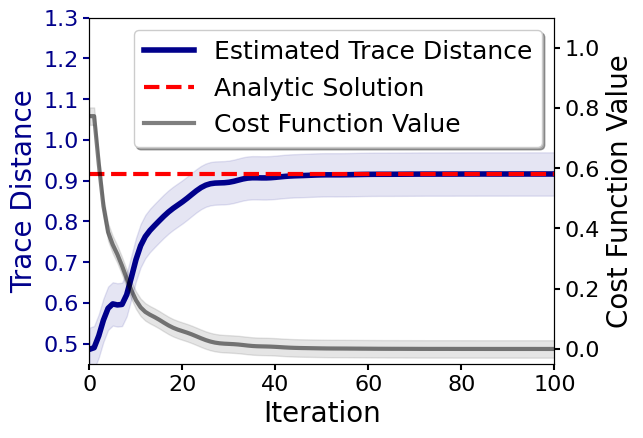}}}
\newcommand\fourfidelityone{\adjustbox{valign=m, vspace=0pt, margin=1ex}{\includegraphics[width=.310\linewidth]{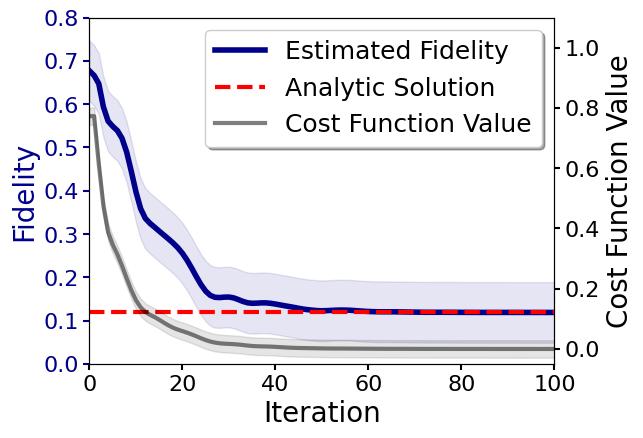}}}

\newcommand\fourvontwo{\adjustbox{valign=m, vspace=0pt, margin=1ex}{\includegraphics[width=.310\linewidth]{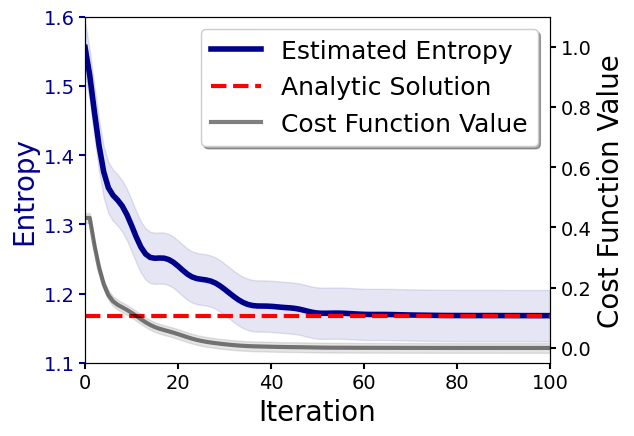}}}
\newcommand\fourrenyitwo{\adjustbox{valign=m, vspace=0pt, margin=1ex}{\includegraphics[width=.310\linewidth]{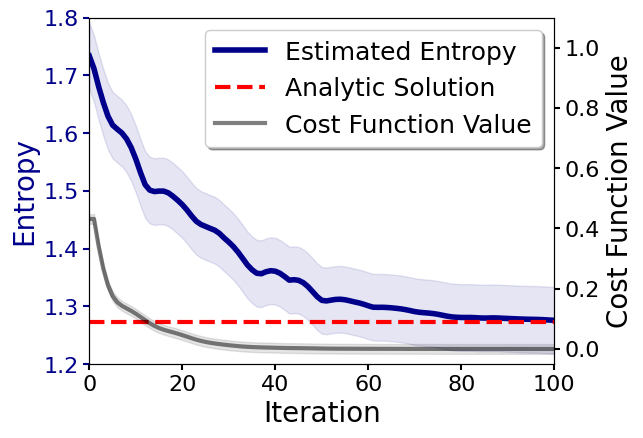}}}
\newcommand\fourtsallistwo{\adjustbox{valign=m, vspace=0pt, margin=1ex}{\includegraphics[width=.310\linewidth]{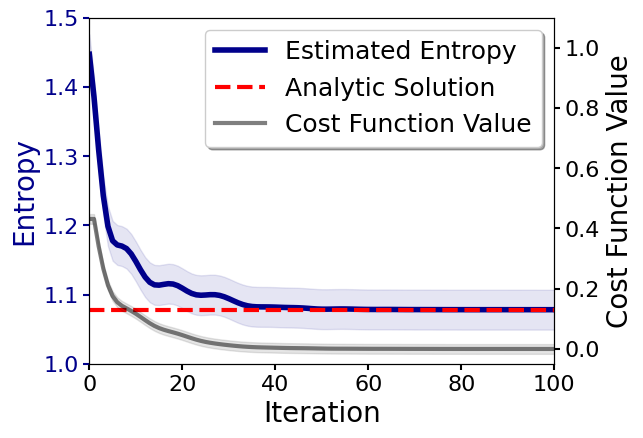}}}
\newcommand\fourtracetwo{\adjustbox{valign=m, vspace=0.pt, margin=1ex}{\includegraphics[width=.310\linewidth]{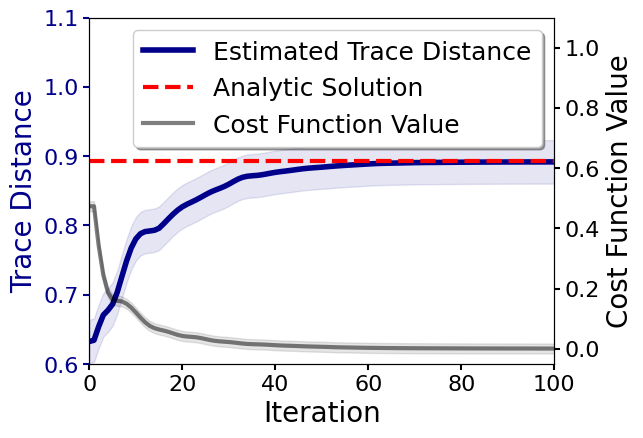}}}
\newcommand\fourfidelitytwo{\adjustbox{valign=m, vspace=0pt, margin=1ex}{\includegraphics[width=.310\linewidth]{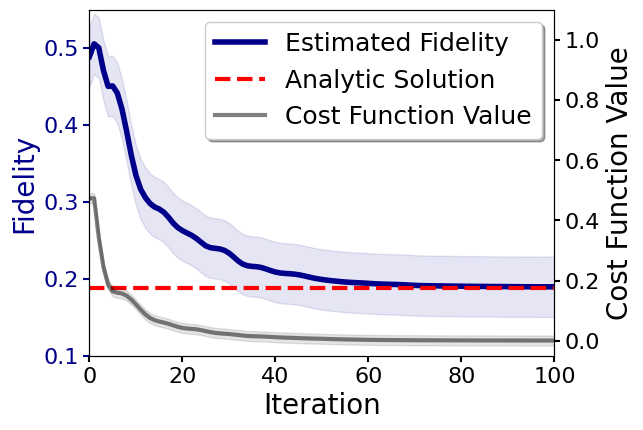}}}

\theoremstyle{plain}
\newtheorem{prop}{Proposition}
\newtheorem{proposition}{Proposition}
\newtheorem{theorem}{Theorem}
\newtheorem{lemma}{Lemma}
\newtheorem{corollary}{Corollary}[proposition]
\newtheorem{remark}{Remark}

\theoremstyle{definition}
\newtheorem{definition}{Definition}

\title{Disentangling quantum neural networks for unified estimation of \\ quantum entropies and distance measures}
\author{Myeongjin Shin}
\thanks{These authors contributed equally to this work.}
\affiliation{School of Computing, KAIST, Daejeon 34141, Korea}

\author{Seungwoo Lee}
\thanks{These authors contributed equally to this work.}
\affiliation{School of Computing, KAIST, Daejeon 34141, Korea}

\author{Junseo Lee}
\thanks{These authors contributed equally to this work.}
\affiliation{Quantum AI Team, Norma Inc., Seoul 04799, Korea}

\author{Mingyu Lee}
\affiliation{Department of Computer Science and Engineering, Seoul National University, Seoul 08826, Korea}

\author{Donghwa Ji}
\affiliation{College of Liberal Studies, Seoul National University, Seoul 08826, Korea}

\author{Hyeonjun Yeo}
\affiliation{Department of Physics and Astronomy, Seoul National University, Seoul 08826, Korea}

\author{Kabgyun Jeong}
\email{kgjeong6@snu.ac.kr}
\affiliation{Research Institute of Mathematics, Seoul National University, Seoul 08826, Korea}
\affiliation{School of Computational Sciences, Korea Institute for Advanced Study, Seoul 02455, Korea}

\date{\today}

\begin{abstract}
The estimation of quantum entropies and distance measures, such as von Neumann entropy, R\'{e}nyi entropy, Tsallis entropy, trace distance, and fidelity-induced distances such as the Bures distance, has been a key area of research in quantum information science. In our study, we introduce the disentangling quantum neural network (DEQNN), designed to efficiently estimate various physical quantities in quantum information. Estimation algorithms for these quantities are generally tied to the size of the Hilbert space of the quantum state to be estimated. Our proposed DEQNN offers a unified dimensionality reduction methodology that can significantly reduce the size of the Hilbert space while preserving the values of diverse physical quantities. We provide an in-depth discussion of the physical scenarios and limitations in which our algorithm is applicable, as well as the learnability of the proposed quantum neural network.
\end{abstract}

\maketitle

\tableofcontents


\section{Introduction} 
Quantum entropies and distance measures play a fundamental role in understanding the properties of quantum states. Essentially, quantum entropies quantify how mixed a quantum state is, while quantum distance measures represent the distinguishability between two quantum states. Estimating these quantities is crucial for understanding quantum states and systems.

In this paper, our objective is to estimate quantum entropies and distances within quantum information theory, encompassing von Neumann entropy~\cite{neumann1927thermodynamik, shannon1948mathematical}, quantum R\'{e}nyi entropy~\cite{renyi1961measures, umegaki1954conditional}, quantum Tsallis entropy~\cite{tsallis1988possible}, trace distance~\cite{helstrom1969quantum, helstrom1967detection}, and fidelity~\cite{jozsa1994fidelity, uhlmann1976transition}. The estimation of these quantities has garnered significant interest, resulting in numerous papers~\cite{acharya2019measuring, acharya2020estimating, subramanian2021quantum, wang2024new, gilyen2019distributional, wang2022quantum, wang2023fast, wang2023quantum, gilyen2019quantum, gilyen2022improved, liu2021solving, chen2021variational, shin2024estimating, goldfeld2024quantum, lee2024estimating}. For current approach, we need to estimate all the quantities individually. Consequently, different algorithmic implementations are necessary for each quantity.

Initial approaches~\cite{acharya2019measuring, acharya2020estimating} used quantum state tomography~\cite{cramer2010efficient} for the estimation of quantum entropies. However, quantum state tomography allows for the reconstruction of the quantum state, enabling the calculation of quantum entropies and distance measures. These approaches are inefficient due to high complexity. To overcome this issue, several algorithmic methods have been proposed, including the use of quantum query access models~\cite{subramanian2021quantum, wang2024new, gilyen2019distributional, wang2022quantum, wang2023fast}, Fourier transformation and Taylor series~\cite{wang2023quantum}, and block encodings~\cite{gilyen2019quantum, gilyen2022improved}.

Quantum query access models address the high (computational) complexity issue, requiring only polynomial rank query access. However, these methods necessitate the quantum circuit that generates the quantum state, which is impractical in many situations. Additionally, the effectiveness of the quantum query model construction remains an open question~\cite{wang2022quantum}. Without the need for quantum query access, the algorithm using Fourier transformation and Taylor series~\cite{wang2023quantum} only requires identical copies of the state. However, this algorithm is ineffective when the smallest non-zero eigenvalue of the quantum state is exponentially small, which is common in many cases. The block encoding algorithm is suggested for fidelity estimation~\cite{gilyen2022improved}, needing only identical copies of the quantum state and having complexity polynomial to the rank. Each algorithmic method only estimates individual quantities, necessitating different approaches for each quantum entropy and distance measure, leading to inefficiency.

Recently, heuristic methods using variational quantum circuits have been proposed~\cite{liu2021solving, chen2021variational, shin2024estimating, goldfeld2024quantum, lee2024estimating}. These methods utilize variational formulas, where the lower or upper bounds of each formula correspond to a quantity value. To optimize the variational formula and determine the lower or upper bounds, variational quantum circuits are employed. The complexity of these methods is not known, but numerical simulations show they are efficient in many cases. The advantage lies in the ease of implementing this algorithm in software, as only the variational formula needs to be altered to estimate other quantities. The ansatz and parameters in variational quantum circuits must be adjusted for each quantity, preventing the reuse of optimized parameters. Consequently, the number of training instances is proportional to the number of quantities to be estimated, leading to an increase in copy complexity. Additionally, methods employing variational quantum circuits face a critical issue known as barren plateaus~\cite{mcclean2018barren, marrero2021entanglement}. On barren plateaus, the quantum circuit is untrainable because the gradient exponentially vanishes with the size of the system. Variational quantum circuits that use global cost functions are vulnerable to barren plateaus~\cite{cerezo2021cost}. Recent methods that use variational quantum circuits~\cite{liu2021solving, chen2021variational, shin2024estimating, goldfeld2024quantum, lee2024estimating} suffer from barren plateaus due to global cost functions.

This paper introduces a variational method for the estimation of quantum entropies and distance measures. The main concept is to find a smaller quantum state that preserves the quantum entropies and distance measures using variational quantum circuits. This concept is referred to as a dimension reduction or data compression technique. Quantum autoencoders~\cite{romero2017quantum} were proposed as an example of a dimension reduction technique for efficient data handling. We propose that dimension reduction can be viewed as disentangling quantum states and prove several important characteristics of disentanglement.

We prove that disentangling quantum states preserve quantum entropies and distance measures in smaller quantum states. Thus, the variational quantum circuit is designed to disentangle the quantum states. Our method can be generalized to estimate the quantum entropies of an arbitrary number of states and distance measures of those pairs using only one variational quantum circuit. {For example, to calculate the von Neumann entropy, Rényi entropy, and Tsallis entropy for a single quantum state, there is no need for separate learning processes. By finding just one disentangling unitary operator $U$ corresponding to the quantum state, all values can be estimated at once, offering an advantage. Specifically, this approach is highly beneficial in common cases where we must repeatedly calculate the entropy to determine the appropriate parameter or which type of generalized entropy that accurately describes the given system \cite{baez2022renyi, tsallis2009nonadditive}. In addition, simultaneously estimating trace distance and fidelity is helpful because they represent different perspectives of the system respectively.} We support our results with numerical simulations and analyze strategies to avoid barren plateaus by using local cost functions.

This paper is structured as follows. Section~\ref{sec:background} explores the definitions and characteristics of quantum entropies and distance measures, with a detailed explanation of the continuity bound. Section~\ref{sec:Disentanglement} discusses methods for disentangling quantum states in the presence of noise and proves that quantum entropies and distance measures are preserved during the disentanglement process. Section~\ref{sec:Main} defines the \textit{disentangling quantum neural networks} (DEQNN) model and explains how quantum information quantities are estimated through the associated cost function. Numerical simulation results are also presented in Section~\ref{sec:simulation}. Finally, Section~\ref{sec:discussion} provides a summary of the overall results and discusses future research directions. Several detailed proofs and mathematical specifics are included in the appendixes.


\section{Quantum Entropies and Distance Measures}\label{sec:background}


The primary purpose of reducing the Hilbert space is to facilitate easier extraction of information from the system. While this reduction may also preserve other important quantities, our primary focus will be on two key quantities: quantum entropies and distance measures. In this section, we briefly introduce these two concepts and explain how our method addresses them.

\subsection{Definitions and Properties}
Entropy serves as a metric for gauging randomness or uncertainty within a system. Quantum entropies represent an extension of this concept within the realm of quantum mechanics, quantifying the informational content inherent in a density matrix. One of the most standard quantum entropies is the von Neumann entropy $S(\rho)=-\T(\rho\log_2\rho)$, which serves as the quantum analog of classical Gibbs entropy.

Under short-range interactions, the von Neumann entropy is extensive and effectively reconstructs classical statistical mechanics within a quantum mechanical framework, similarly to Gibbs entropy \cite{tsallis2022entropy}. This approach allows for the clear definition of important properties, such as the thermodynamic limit, temperature, and non-negative specific heat \cite{bogolubov2010introduction}. Additionally, numerous applications of the von Neumann entropy have emerged in quantum information theory, including its use as an entanglement measure \cite{nielsen2001quantum}, for determining channel capacity \cite{holevo1973bounds}, and in quantum phase detection \cite{kitaev2006topological}.

In addition to the von Neumann entropy $S(\rho)$, there are mathematical generalization, such as the quantum R\'{e}nyi entropy 
\begin{align}
    S_\alpha(\rho)=\frac{1}{1-\alpha}\log\T(\rho^\alpha),~\alpha\in(0,1)\cup(1,\infty),
\end{align}
and the quantum Tsallis entropy
\begin{align}
    S_q(\rho)=\frac{1}{1-q}\left(\T(\rho^q)-1\right),~q\in(0,1)\cup(1,\infty).
\end{align}
These quantum entropies are analogs of the classical R\'{e}nyi and the Tsallis entropies. As the parameters $\alpha$ and $q$ approach 1, both $S_{\alpha}(\rho)$ and $S_q(\rho)$ converge to the von Neumann entropy $S(\rho)$, as expressed by
\begin{equation}
\lim\limits_{\alpha \to 1} S_{\alpha}(\rho) = S(\rho),~\textnormal{and}~\lim\limits_{q \to 1} S_q(\rho) = S(\rho).
\end{equation}
This insight reveals an interconnectedness among quantum entropies, facilitating a unified approach for estimation across different measures. 

From a physical perspective, Tsallis entropy was proposed as a generalization of Gibbs entropy for non-extensive physics \cite{tsallis2009nonadditive, tsallis2022entropy}. The non-extensive statistical mechanics framework based on Tsallis entropy has been applied across various fields, particularly in studies of long-range interacting systems \cite{sheykhi2018modified, ourabah2022generalized} and entangled systems \cite{caruso2008nonadditive}. In contrast, Rényi entropy originated from information theory \cite{renyi1961measures} as a means to generalize entropies mathematically. The Rényi entropy also has significant physical interpretations, serving as a free energy \cite{baez2022renyi} and as a measure of entanglement \cite{wang2016entanglement, brydges2019probing}.

However, there are some theoretical issues with these generalized entropies. The most significant is that both Rényi and Tsallis entropies lack Lesche stability \cite{abe2002stability, lutsko2009tsallis}, meaning they can be discontinuous in response to small changes in the distribution. This instability prevents these entropies from being physically well-defined. Additional issues remain, such as those related to the maximum entropy principle \cite{jizba2019maximum}, the lack of a clear thermodynamic limit when non-extensivity is present, interpreting Rényi entropy as a free energy as $\alpha\rightarrow 1$, and the occurrence of negative specific heat in the context of Tsallis entropy \cite{boon2011nonextensive}.

We clarify that our algorithm does not aim to unify generalized statistical mechanics or resolve theoretical issues within generalized entropies. Bearing these limitations in mind, we cautiously propose that our Hilbert space reduction strategy is still valuable and can be applied to estimate quantum Rényi and Tsallis entropies, given numerous studies across diverse fields that utilize these generalized entropies \cite{laflorencie2016quantum, brydges2019probing, wen2019measuring}. Notably, while the lack of Lesche stability cannot be entirely circumvented, we have verified that for the random distribution inputs used in our numerical simulations, the estimated values of generalized entropies remained stable under small perturbations.

Distance measures serve as a fundamental instrument in assessing the distinguishability between two quantum states, with potential applicability extending to quantum channels. Key metrics within this domain include the trace distance $T(\rho, \sigma)$ and the Hilbert-Schmidt distance $D_{HS}(\rho, \sigma)$. These distance measures, denoted by $D$, are characterized by the following properties:
\begin{itemize}
    \item Positivity: $D(\rho, \sigma) \geq 0$
    \item Symmetry: $D(\rho, \sigma) = D(\sigma, \rho)$
    \item Triangle inequality: $D(\rho, \gamma) + D(\gamma, \sigma) \geq D(\rho, \sigma)$.
\end{itemize} 
While the Hilbert-Schmidt distance can indeed be efficiently computed using the swap test on a quantum computer~\cite{garcia2013swap}, the focus of our paper centers on the trace distance, given its practical relevance and applicability.

The fidelity $F(\rho, \sigma)$ serves as a widely employed measure for assessing the distance between density operators, reflecting the degree of similarity between quantum states. It is important to note that fidelity, while not meeting the criteria of a metric due to its failure to satisfy the triangle inequality, can still be utilized to establish metrics within the space of density matrices. Examples include Bures angle $D_A(\rho, \sigma) = \arccos F(\rho, \sigma)$ and Bures distance $D_B(\rho, \sigma) = \sqrt{2-2\sqrt{F(\rho, \sigma)}}$. Consequently, estimating fidelity facilitates the estimation of these metrics. All aforementioned measures, including trace distance, fidelity-derived distances such as Bures distance and angle, characterize the distance between quantum states, thus demonstrating {connection}. This interconnectedness opens avenues for a unified approach to their estimation.

\begin{table}
\caption{\textbf{Best known sample complexity for estimating quantum information quantities.} This table presents key quantum information quantities—von Neumann entropy $S(\rho)$, R\'{e}nyi entropy $S_\alpha(\rho)$ (where $\alpha$ is a positive, non-integer value), Tsallis entropy $S_q(\rho)$ (omitted in this table), trace distance $T(\rho, \sigma)$, and fidelity $F(\rho, \sigma)$—and their best-known sample complexities for estimation. Here, $N=2^n$ denotes the dimension of the $n$-qubit quantum state, and $\epsilon$ is the desired additive error. A more detailed summary of the contents in this table is well explained in the paper by Wang et al~\cite{wang2024new}. \\ $^\dagger$Note that for measures like trace distance and fidelity, the sample complexity for low-rank quantum states may be expressed in a rank-dependent manner, specifically $\mathcal{O}(r/\epsilon^2)$ and $\mathcal{O}(r/\epsilon)$, where $r$ is an upper bound for the rank of the quantum states.}
\begin{ruledtabular}
\renewcommand{\arraystretch}{1.2}
\begin{tabular}{c | c}
Quantities & (Best Known) Sample Complexity \\
\hline
$S(\rho)=-\T(\rho\log\rho)$ &  $\mathcal{O}(N^2/\epsilon^2)$ \\
\hline
$S_\alpha(\rho)=\frac{1}{1-\alpha}\log\T(\rho^\alpha)$ & $\mathcal{O}\left((N/\epsilon)^{\max\{2/\alpha, 2\}}\right)$ \\ 
\hline
$T(\rho, \sigma) = \frac{1}{2}\T\left(|\rho-\sigma|\right)$ & $\mathcal{O}(N/\epsilon^2)^\dagger$ \\
\hline
$F(\rho,\sigma) = \left(\T\sqrt{\sqrt{\rho}\sigma\sqrt{\rho}}\right)^2$ & $\mathcal{O}(N/\epsilon)^\dagger$
\end{tabular}
\end{ruledtabular}
\label{table:est_comp}
\end{table}

There are two main approaches for quantum algorithms in estimating quantum physical quantities. The first is the \textit{quantum sample access} model, where copies of the quantum state are provided and can be used directly. The other is the \textit{purified quantum query access} model, in which a mixed quantum state can be accessed in oracle form to prepare its purification. In this paper, we consider the quantum sample access model, where the number of copies of the quantum state required for the algorithm to operate within a desired error is referred to as the \textit{sample complexity.}

For von Neumann entropy and Rényi entropy, algorithms with sample complexities of $\mathcal{O}(N^2/\epsilon^2)$ and $\mathcal{O}\left((N/\epsilon)^{\max\{2/\alpha, 2\}}\right)$, respectively, have been proposed~\cite{acharya2020estimating}, though these algorithms are based on weak Schur sampling, making it non-trivial to reduce the sample complexity in a rank-dependent manner. Recently, a Tsallis entropy $S_q(\rho)$ estimation algorithm for $n$-qubit quantum states with rank-dependent sample complexity was proposed~\cite{liu2024estimating}. It was shown that if $q \geq 1 + \Omega(1)$, the lower bound of the sample complexity is $\tilde{O}(1/\epsilon^{3+\frac{2}{q-1}})$, while if $1 \leq q \leq 1 + \frac{1}{n - 1}$, the lower bound of the sample complexity is $\tilde{O}(r^2/\epsilon^5)$. Additionally, given copies of a quantum state $\rho$, the sample complexity of quantum algorithms for calculating the trace distance and fidelity between two quantum states through quantum state certification~\cite{buadescu2019quantum}—the task of determining if the quantum state is identical to a known quantum state $\sigma$ or $\epsilon$-far—has been shown to be $\mathcal{O}(N/\epsilon^2)$ and $\mathcal{O}(N/\epsilon)$, respectively. In this case, it has also been proven that the sample complexity can be improved in a rank-dependent manner for low-rank quantum states. These sample complexities are summarized in Table~\ref{table:est_comp}. This demonstrates that typical quantum physical quantity estimation algorithms have sample complexities proportional to the Hilbert-space dimension. Our DEQNN has significance as it is the first to propose a {variational quantum algorithm~\cite{cerezo2021variational}} approach that reduces the Hilbert space while preserving these physical quantities.

\begin{table*}
\caption{\textbf{Summary of definitions and continuity bounds for quantum entropies and distance measures.} This table presents continuity bounds for various quantum information quantities, where $T_{\rho} = T(\rho_1, \rho_2)$, $T_{\sigma} = T(\sigma_1, \sigma_2)$, $r$ satisfies $r \geq \max\{\text{rank}(\rho), \text{rank}(\sigma)\}$, and $m$ satisfies $m \rho_1 \leq \sigma_1$, $m \rho_2 \leq \sigma_2$.}
\begin{ruledtabular}
\renewcommand{\arraystretch}{1.2}
\begin{tabular}{c | c | c}
Quantities  & Continuity Bounds & Ref.  \\
\hline
$S(\rho)=-\T(\rho\log\rho)$ & $|S(\rho_1)-S(\rho_2)| \leq T_{\rho}\log\left(r-1\right) -T_{\rho}\log T_{\rho} - (1-T_{\rho})\log\left(1-T_{\rho}\right)$ & \cite{audenaert2007sharp} \\
\hline
$S_\alpha(\rho)=\frac{1}{1-\alpha}\log\T(\rho^\alpha),~\alpha\in(0,1)\cup(1,\infty)$  & 
$|S_{\alpha}(\rho_1)-S_{\alpha}(\rho_2)| \leq$
$\begin{cases}
    \frac{1}{1-\alpha}((1-T_{\rho})^{\alpha}-1+(r-1)^{1-\alpha}T_{\rho}^\alpha), & \alpha\in(0,1)\\
    \frac{2\alpha}{\alpha-1}r^{\alpha-1}T_{\rho}, & 
    \alpha\in(1,\infty) \\
\end{cases}$ & \cite{hanson2017tight, chen2017sharp}\\ 
\hline
$S_q(\rho)=\frac{1}{1-q}(\T(\rho^q)-1),~q\in(0,1)\cup(1,\infty)$ &  $|S_q(\rho_1)-S_q(\rho_2)| \leq$  
$\begin{cases}
    \frac{1}{1-q}((1-T_{\rho})^q-1+(r-1)^{1-q}T_{\rho}^q), & q\in(0,1)\\
    \frac{2q}{q-1}T_{\rho}, & q\in(1,\infty) \\
\end{cases}$ & \cite{raggio1995properties, chen2017sharp}\\
\hline
$T(\rho, \sigma) = \frac{1}{2}\T\left(|\rho-\sigma|\right)$ & $|T(\rho_1, \sigma_1)-T(\rho_2, \sigma_2)| \leq T_{\rho} + T_{\sigma}$ & \cite{nielsen2001quantum}\\
\hline
$F(\rho,\sigma) = \left(\T\sqrt{\sqrt{\rho}\sigma\sqrt{\rho}}\right)^2$ & $|F(\rho_1, \sigma_1)-F(\rho_2, \sigma_2)| \leq \arccos (1-T_{\rho})^2 + \arccos (1-T_{\sigma})^2$ & \cite{nielsen2001quantum}
\end{tabular}
\end{ruledtabular}
\label{table:cont_bound}
\end{table*}

\subsection{Continuity Bounds}

Continuity bounds demonstrate that when quantum states are in proximity, the associated quantities also exhibit proximity. Typically, continuity bounds for quantum entropies and distance measures are articulated in relation to the trace distance. The detailed representations of continuity bounds, expressed in terms of trace distance, are presented in Table~\ref{table:cont_bound}.

The continuity bounds for von Neumann, R\'{e}nyi, and Tsallis entropy are well-documented. In this work, we refine the bounds for R\'{e}nyi and Tsallis entropy, specifically in terms of the rank $r = \text{rank}(\rho)$. Here, the rank $r$ of a density matrix $\rho$ denotes the number of non-zero eigenvalues in $\rho$.

In contrast, fidelity does not directly adhere to the triangle inequality. However, it is established that fidelity satisfies the relationship~\cite{nielsen2001quantum}:
\begin{equation}\label{eq:2}
\arccos F(\rho, \sigma) + \arccos F(\sigma, \gamma) \geq \arccos F(\rho, \gamma).
\end{equation}
Further details on the proofs of continuity bounds is provided in \textbf{Appendix~\ref{sec:proofcontbound}}.

We define $F_1 = \{S, S_{\alpha}, S_q\}$ and $F_2 = \{T, F\}$ for convenience. This notation will be used in later sections. Then, the left side of the continuity bound with respect to trace distance can be expressed using $f_1 \in F_1$ and $f_2 \in F_2$. The quantities in $F_1$ represent the information of one quantum state, including von Neumann entropy, R\'{e}nyi entropy, and Tsallis entropy. The quantities in $F_2$ represent the information between two quantum states, including trace distance and fidelity. If $T(\rho_1, \rho_2),~T(\sigma_1, \sigma_2) < \epsilon$, by using the continuity bounds of $F_1$ we deduce that
\begin{equation}
|f_1(\rho_1) - f_1(\rho_2)| < Cr^a\epsilon^b,
\end{equation}
and, by using the continuity bounds of $F_2$, we address that
\begin{equation}
|f_2(\rho_1, \sigma_1) - f_2(\rho_2, \sigma_2)| < Cr^a\epsilon^b,
\end{equation}
for some constants $C, a, b$, which are determined by $f_1$ and $f_2$. Here, $r = \max\{\text{rank}(\rho_1), \text{rank}(\rho_2)\}$. Further details are provided in \textbf{Appendix~\ref{sec:prooflem1}}.

{Note that unlike the von Neumann entropy $S(\rho)$ and the Tsallis entropy $S_q(\rho)$, the R\'{e}nyi entropy $S_\alpha(\rho)$ is not Lesche-stable, meaning it can not be reliably measured in thermodynamic limit \cite{abe2002stability}. However, as observed from the continuity bound above, $S_\alpha(\rho)$ can be reliably measured in finite systems, which is the case we are addressing in this paper.}


\section{Disentanglement with Error}\label{sec:Disentanglement}
In general, applying a unitary transformation can transform an unknown quantum state into a tensor product of quantum states, a process known as \textit{disentanglement}. Our goal is to find a unitary operator $U$ such that $U\ket{\psi}_{AB} = \ket{0}_A \otimes \ket{\psi}_B^*$ for a pure state $\ket{\psi}_{AB}$ and $U\rho_{AB} U^{\dagger} = \ket{0}\bra{0}_A \otimes \rho^*_B$ for a mixed state $\rho_{AB}$. Disentanglement can be used for quantum state tomography~\cite{huang2024learning, yao2023quantum}. Note that for $r=\text{rank}(\rho_{AB})$, by choosing the dimension of system $B$ greater or equal to $r$, there always exists a unitary $U$ such that $U\rho_{AB} U^{\dagger} = \ket{0}\bra{0}_A \otimes \rho^*_B$ for some $\rho^*_B$. (The asterisk in the superscript of the quantum states indicates notation to distinguish between the partial trace of the composite system yielding the subsystem and the subsystem obtained after applying the unitary operator and taking the partial trace.)

However, exact unitaries are often elusive, so we define the disentanglement process with an error tolerance $\epsilon$ and explore the characteristics of these unitary transformations.

\begin{definition}\label{def1}
For an $n$-qubit mixed state $\rho_{AB}$, a unitary transformation $U$ that satisfies
\begin{equation}
\text{Tr}\left(U\rho_{AB} U^{\dagger} \ket{0}\bra{0}_A \otimes I_B\right) \geq 1-\epsilon, 
\end{equation}
is termed a {type-\RN{1} $\epsilon$-approximate disentangling unitary} for $\rho_{AB}$. If $\epsilon = 0$, then $U$ is called a type-\RN{1} perfect disentangling unitary for $\rho_{AB}$. The conditions for the existence of such a perfect disentangling unitary are presented in \textbf{Proposition~\ref{prop1}}.
\end{definition}

Here, $\epsilon$ represents the disentanglement error for $\rho$ using $U$. Such a unitary $U$ can be constructed through the utilization of a parametrized quantum circuit~\cite{benedetti2019parameterized}. Figure~\ref{fig:dis} illustrates this process schematically.

\begin{figure}[h]
    \includegraphics[width=8.5cm]{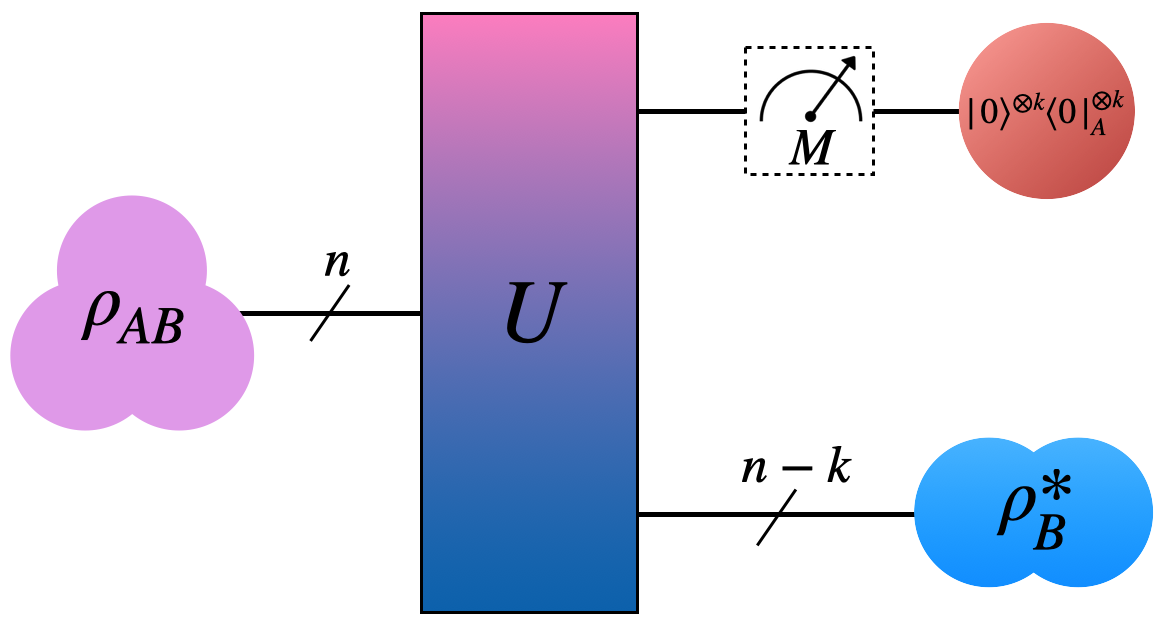}
    \caption{\textbf{Type-\RN{1} $\epsilon$-approximate disentangling unitary $U$ for a quantum state $\rho_{AB}$.} System $A$ becomes a pure state, and the remaining system $B$ is learned to have similar characteristics to the original mixed system $AB$.}
    \label{fig:dis}
    \centering
\end{figure}

And the following is an extended definition of disentanglement between two quantum states, incorporating distance and fidelity, based on \textbf{Definition~\ref{def1}}. And Figure~\ref{fig:dis2} illustrates this process schematically.
\begin{definition}\label{def2}
For $n$-qubit mixed state $\rho_{AB}$ and $\sigma_{AB}$, a \textit{single} unitary transformation $U$ that satisfies \textit{both} of the following conditions:
\begin{align}
\text{Tr}\left(U\rho_{AB} U^{\dagger} \ket{0}\bra{0}_A \otimes I_B\right) \geq 1-\epsilon,~\text{and}\\
\text{Tr}\left(U\sigma_{AB} U^{\dagger} \ket{0}\bra{0}_A \otimes I_B\right) \geq 1-\epsilon,
\end{align}
is termed a type-\RN{2} $\epsilon$-approximate disentangling unitary for $\rho_{AB}$ and $\sigma_{AB}$. If $\epsilon = 0$, then $U$ is called a type-\RN{2} perfect disentangling unitary for $\rho_{AB}$ and $\sigma_{AB}$. The conditions for the existence of such a perfect disentangling unitary are presented in \textbf{Proposition~\ref{prop2}.}
\end{definition}

\begin{figure}[h]
    \includegraphics[width=8.5cm]{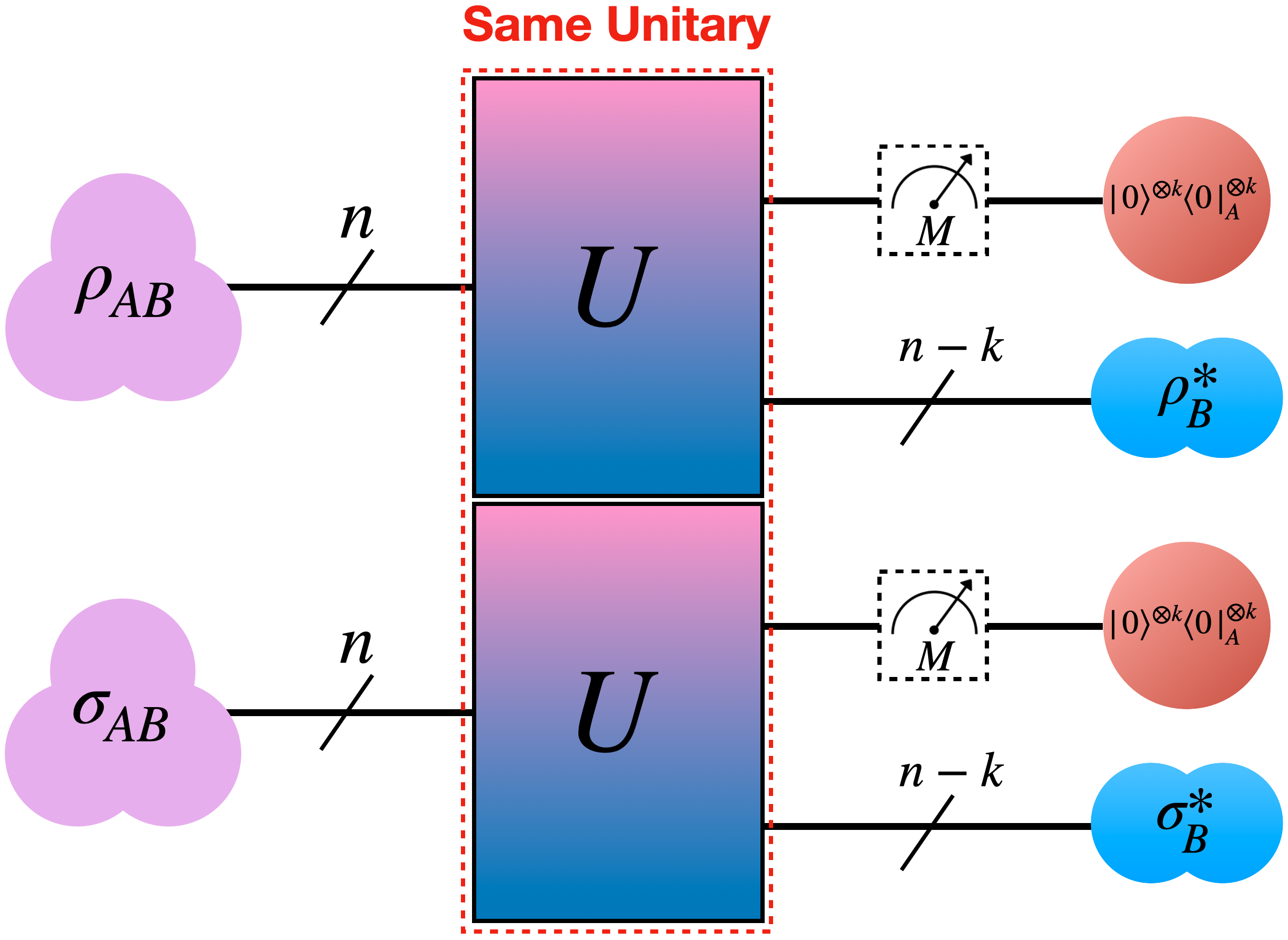}
    \caption{\textbf{Type-\RN{2} $\epsilon$-approximate disentangling unitary $U$ for two quantum states, $\rho_{AB}$ and $\sigma_{AB}$.} Note that the same unitary operation $U$ is applied to both quantum states.}
    \label{fig:dis2}
    \centering
\end{figure}

When the disentanglement error $\epsilon$ is small, it follows that the discrepancy between $\rho_{AB}^* = U\rho_{AB} U^{\dagger}$ and $\ket{0}\bra{0}_A\otimes \rho_B^*$ is also small, where $\rho_B^* = \T_A(U\rho_{AB} U^{\dagger})$. Specifically, the trace distance between $U\rho_{AB} U^{\dagger}$ and $\ket{0}\bra{0}_A\otimes \rho_B^*$ is minimal, given by \textbf{Lemma~\ref{lemma:td}}.
\begin{lemma}\label{lemma:td}
Consider a type-I $\epsilon$-approximate disentangling unitary $U$ for $\rho_{AB}$. The following property holds:
\begin{equation}
T(U\rho_{AB} U^{\dagger}, \ket{0}\bra{0}_A\otimes\rho_B^*) \leq 2\sqrt{r\epsilon},
\end{equation}
where $\rho_B^* = \T_A(U\rho_{AB} U^{\dagger})$ and $r=\textnormal{rank}(\rho_{AB})$.
\end{lemma}
\begin{proof}
See the details of the proof in \textbf{Appendix~\ref{sec:prooflem1}}. 
\end{proof}

The above lemma shows that by minimizing $\epsilon$, the information loss incurred by partial trace can also be minimized. With this perspective, we can infer that disentanglement with error preserves quantum entropies and distance measures. This actually serves as a key component of this paper, which is the main theorem stated below.

\begin{table*}
\caption{\textbf{Expressions of \textbf{Theorems~\ref{thm:1}} and \textbf{\ref{thm:2}} for specific quantum entropies and distance measures.} This table presents the constants $C, a$, and $b$ for specific quantum entropies and distance measures that satisfy inequalities~(\ref{eq:thm1}) and~(\ref{eq:thm2}) as proposed in \textbf{Theorems~\ref{thm:1}} and \textbf{\ref{thm:2}}. These expressions confirm that the values of entropies and distance measures are preserved even when disentanglement is applied.}
\begin{ruledtabular}
\renewcommand{\arraystretch}{1.2}
\begin{tabular}{@{\hspace{3em}} c @{\hspace{3em}} | @{\hspace{3em}} c @{\hspace{3em}}}
Quantities & Inequalities  \\
\hline
$S(\rho)=-\T(\rho\log\rho)$ & $|S(\rho_{AB})-S(\rho_B^*)| < 2r^{\frac{3}{4}}\epsilon^{\frac{1}{4}}$ \\
\hline
$S_\alpha(\rho)=\frac{1}{1-\alpha}\log\T(\rho^\alpha),~\alpha\in(0,1)\cup(1,\infty)$  & 
$|S_{\alpha}(\rho_{AB})-S_{\alpha}(\rho_B^*)| <
\begin{cases}
    \frac{2^{\alpha}}{1-\alpha}r^{1-\frac{\alpha}{2}} \epsilon^{\frac{\alpha}{2}}, & \alpha\in(0,1)\\
    \frac{4\alpha}{\alpha-1}r^{\alpha-\frac{1}{2}}\epsilon^{\frac{1}{2}}, & 
    \alpha\in(1,\infty)
\end{cases}$\\ 
\hline
$S_q(\rho)=\frac{1}{1-q}(\T(\rho^q)-1),~q\in(0,1)\cup(1,\infty)$ & $|S_q(\rho_{AB})-S_q(\rho_B^*)| <
\begin{cases}
    \frac{2^q}{1-q}r^{1-\frac{q}{2}} \epsilon^\frac{q}{2}, & q\in(0,1)\\
    \frac{4q}{q-1}r^\frac{1}{2}\epsilon^\frac{1}{2}, & 
    q\in(1,\infty)
\end{cases}$\\
\hline
$T(\rho, \sigma) = \frac{1}{2}\T\left(|\rho-\sigma|\right)$ & $|T(\rho_{AB}, \sigma_{AB})-T(\rho_B^*, \sigma_B^*)| < 4r^\frac{1}{2}\epsilon^\frac{1}{2}$\\
\hline
$F(\rho,\sigma) = \left(\T\sqrt{\sqrt{\rho}\sigma\sqrt{\rho}}\right)^2$ & $|F(\rho_{AB}, \sigma_{AB})-F(\rho_B^*, \sigma_B^*)| < 2\pi r^\frac{1}{4}\epsilon^{\frac{1}{4}}$
\end{tabular}
\end{ruledtabular}
\label{table:ineq}
\end{table*}

\subsection{Disentanglement Preserves Quantum Entropies}
\begin{theorem}\label{thm:1}
For $n$-qubit state $\rho_{AB}$ and $f_1 \in F_1$, the following property holds:
\begin{equation}\label{eq:thm1}
|f_1(\rho_{AB})-f_1(\rho_B^*)| < Cr^a\epsilon^b,
\end{equation}
for some constant $C, a, b$ which are specified in Table~\ref{table:ineq}. Here, $U$ is the type-I $\epsilon$-approximate disentangling unitary for $\rho_{AB}$, and $\rho_B^* = \T_A(U\rho_{AB}U^{\dagger})$.
\end{theorem}
\begin{proof}
See the details of the proof in \textbf{Appendix~\ref{sec:proofthms}}. 
\end{proof}

This theorem shows that the von Neumann, R\'{e}nyi, and Tsallis entropies of $\rho_{AB}$ can be preserved in the partially traced state $\rho_B^*$ using disentanglement. Therefore, with the type-\RN{1} $\epsilon$-approximate disentangling unitary $U$, the quantum entropies of $\rho_{AB}$ can be calculated using $\rho_B^*$, which has fewer qubits and thus less complexity. \textbf{Theorem~\ref{thm:1}} can be proved using \textbf{Lemma~\ref{lemma:td}} and the continuity bounds of quantum entropies.

\subsection{Disentanglement Preserves Quantum Distance Measures}
\begin{theorem}\label{thm:2}
For $n$-qubit states $\rho_{AB}, \sigma_{AB}$ and $f_2 \in F_2$, the following property holds:
\begin{equation}\label{eq:thm2}
|f_2(\rho_{AB}, \sigma_{AB})-f_2(\rho_B^*, \sigma_B^*)| < Cr^a\epsilon^b,
\end{equation}
for some constant $C, a, b$ which are specified in Table~\ref{table:ineq}. Here, $U$ is the type-II $\epsilon$-approximate disentangling unitary for $\rho_{AB}$ and $\sigma_{AB}$, where $\rho_B^* = \T_A(U\rho_{AB}U^{\dagger}), \sigma_B^* = \T_A(U\sigma_{AB}U^{\dagger})$.
\end{theorem}
\begin{proof}
See the details of the proof in \textbf{Appendix~\ref{sec:proofthms}}. 
\end{proof}

This theorem shows that the trace distance and fidelity of $\rho_{AB}$ and $\sigma_{AB}$ can be preserved in the partially traced states $\rho_B^*$ and $\sigma_B^*$ using disentanglement. Thus, with the type-\RN{2} $\epsilon$-approximate disentangling unitary $U$, we can calculate the quantum distance measures of $\rho_{AB}$ and $\sigma_{AB}$ using $\rho_B^*$ and $\sigma_B^*$, which have fewer qubits and therefore less complexity. Also, \textbf{Theorem~\ref{thm:2}} can be proved using \textbf{Lemma~\ref{lemma:td}} and the continuity bounds of quantum distance measures.

Disentanglement is a special form of dimension reduction technique. There are other dimension reduction techniques, such as quantum variational autoencoders~\cite{khoshaman2018quantum}. It seems obvious that any dimension reduction technique with zero error will preserve quantum entropies and distance measures. However, with an error, it is important to prove that the difference in the quantities does not grow exponentially due to the error. For example, if the difference in the quantities is $2^n\epsilon$ due to the error $\epsilon$ of the dimension reduction technique, we cannot use that technique for estimating the quantities. Therefore, we argue that \textbf{Theorems~\ref{thm:1}} and \textbf{\ref{thm:2}} are important statements.

In the next section, we focus on finding an $\epsilon$-approximate disentangling unitary $U$ using a quantum neural network. Note that if $U$ can disentangle $\rho_{AB}$ to any pure state $\ket{\psi}_A$, not only $\ket{0}_A$, \textbf{Theorems~\ref{thm:1}} and \textbf{\ref{thm:2}} still hold.

\section{Main Result: Disentangling Quantum Neural Networks} \label{sec:Main}
We propose \textit{disentangling quantum neural networks}, namely DEQNN, a quantum neural network that uses disentanglement for estimating quantum entropies and distance measures. DEQNN, which is used to calculate the quantum entropies and distance measures of $\rho_{AB}$ and $\sigma_{AB}$, employs the following cost function:
\begin{equation}
    C(\theta) = 1-\frac{1}{4} \T((\rho_A(\theta)+\sigma_A(\theta))^2),
\end{equation}
where $U(\theta)$ is the unitary representation of the quantum neural network, $\rho_A(\theta) = \T_B(U(\theta)\rho_{AB} U(\theta)^{\dagger})$ and $\sigma_A(\theta) = \T_B(U(\theta)\sigma_{AB} U(\theta)^{\dagger})$. {Note that if a quantum state $\rho$ consists of two separable subsystems, such that $\rho = \rho_i(\theta)\otimes\rho_j(\theta)$, we have $\text{Tr}(S(\rho_i(\theta)\otimes\rho_j(\theta))) = \text{Tr}(\rho_i(\theta)\rho_j(\theta))$, where $S$ is the swap operator defined as $S|\alpha\rangle_A|\beta\rangle_B = |\beta\rangle_A|\alpha\rangle_B$ for any pure states $|\alpha\rangle_A$ and $|\beta\rangle_B$. This is easily proved using the spectral decomposition of $\rho_i(\theta)$ and $\rho_j(\theta)$.~\cite{ekert2002direct}. The quantum circuit for performing the swap test is shown in Figure~\ref{fig:swap}.}

\begin{figure}[h]
    \includegraphics[width=7.5cm]{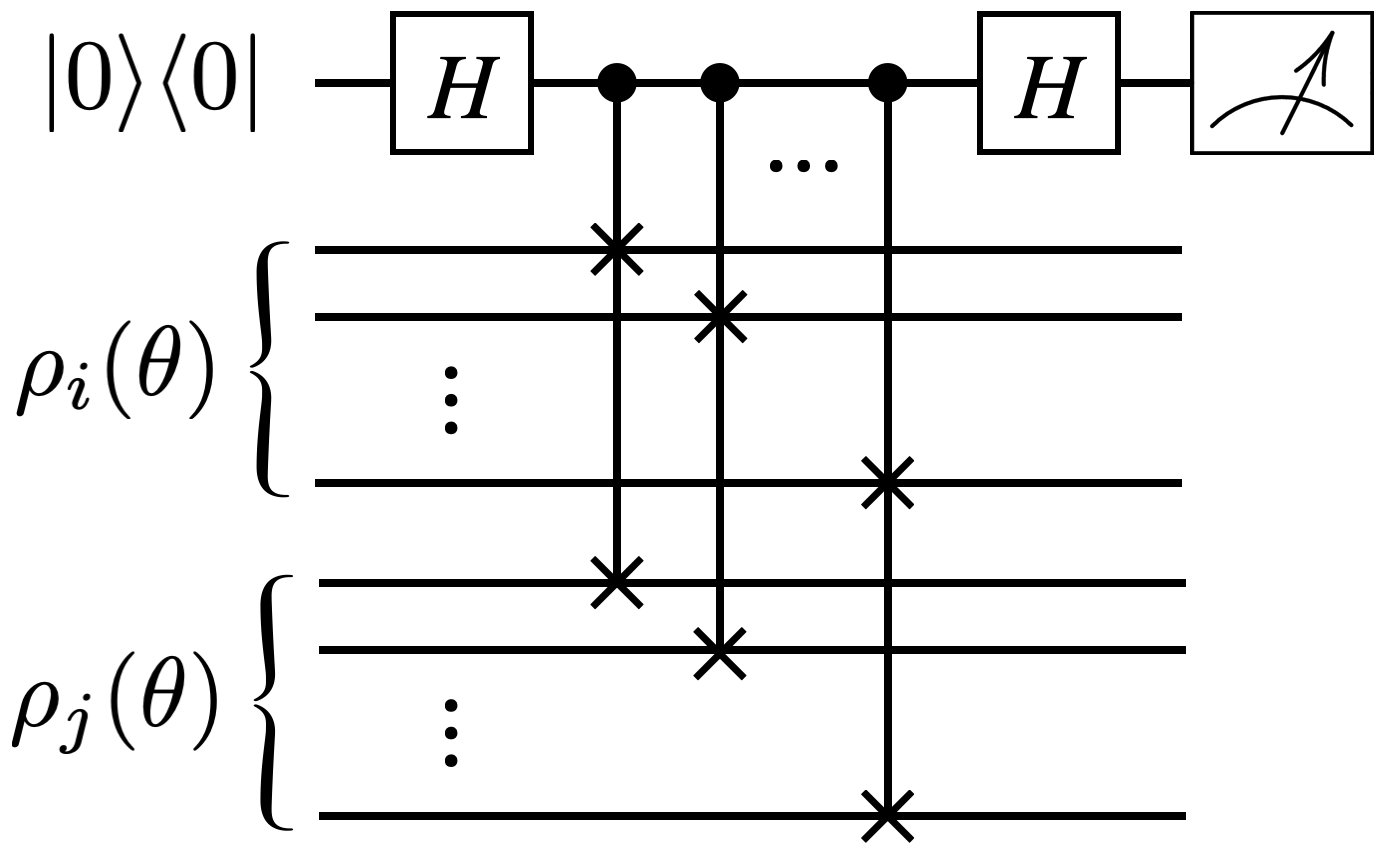}
    \caption{\textbf{Quantum circuit for the swap test.} The result obtained from this circuit is $\text{Tr}(\rho_i(\theta)\rho_j(\theta))$, and this value can be used to design the loss function of the DEQNN.}
    \label{fig:swap}
    \centering
\end{figure}

Hence, the cost function of DEQNN can be written as:
\begin{equation}\label{eq:swap}
    C(\theta) = 1-\frac{1}{4} \sum_{1\leq i,j\leq 2}\T\left(S_A (\rho_i(\theta)\otimes\rho_j(\theta))\right),
\end{equation}
where $\rho_1 = \rho_{AB}$, $\rho_2 = \sigma_{AB}$, $\rho_i(\theta) = U(\theta) \rho_i U^{\dagger}(\theta)$, and $S_A$ is the swap operator acting on system $A$. The cost function $C(\theta)$ can be calculated using the swap test~\cite{garcia2013swap}. When $\theta$ is fully optimized, we can assume that the quantum entropies of $\rho_B(\theta)$, $\sigma_B(\theta)$, and $\rho_{AB}$, $\sigma_{AB}$ are almost the same. We prove this statement in the next section.

\subsection{Cost Function Analysis}\label{sec:Cost}
The design of the cost function was inspired by the concept of purity. Purity is a measure of how mixed a state is. The purity of $\rho$ is defined as $\T(\rho^2)$. If the purity of a state is close to $1$, it can be approximated as a pure state.

\begin{lemma}\label{lemma:cost}
Suppose $C(\theta) < \epsilon$. There exists a pure state $\ket{\psi}$ that satisfies
\begin{equation}\label{eq:11}
\bra{\psi}\rho_A(\theta)\ket{\psi} \geq 1-O(\epsilon), \;\;\textnormal{and}
\end{equation}
\begin{equation}
\bra{\psi}\sigma_A(\theta)\ket{\psi} \geq 1-O(\epsilon).
\end{equation}
\end{lemma}
\begin{proof}
See the details of the proof in \textbf{Appendix~\ref{sec:prooflem2}}. 
\end{proof}

The above lemma implies that optimizing $C(\theta)$ leads $\rho_A(\theta)$ and $\sigma_A(\theta)$ to an equivalent pure state. It also drives $U(\theta)$ towards a type-\RN{2} disentangling unitary for $\rho_{AB}$ and $\sigma_{AB}$, with only a small error. After optimizing $\theta$ by minimizing $C(\theta)$, the values of quantum entropies and distance measures are preserved in $\rho_B(\theta)$ and $\sigma_B(\theta)$. Thus, we call system $A$ the discarded system and system $B$ the preserved system, respectively.

\begin{theorem}\label{thm:3}
For $n$-qubit states $\rho_{AB}$ and $\sigma_{AB}$, and $f_1 \in F_1$, $f_2 \in F_2$, suppose $C(\theta) < \epsilon$, then
\begin{equation}
|f_1(\rho_{AB})-f_1(\rho_B(\theta))| < Cr^a\epsilon^b, \;\;\textnormal{and}
\end{equation}
\begin{equation}
|f_2(\rho_{AB}, \sigma_{AB})-f_2(\rho_B(\theta), \sigma_B(\theta))| < Cr^a\epsilon^b,
\end{equation}
for some constants $C$, $a$, and $b$. Note that $a$ and $b$ are the same values as in Table~\ref{table:ineq}.
\end{theorem}
\begin{proof}
This is obvious by \textbf{Lemma~\ref{lemma:cost}}, and \textbf{Theorems~\ref{thm:1}} and \textbf{\ref{thm:2}}.
\end{proof}

Choosing how to split systems $A$ and $B$ is crucial for optimizing performance. In \textbf{Propositions~\ref{prop1} and~\ref{prop2}}, we established the conditions under which a perfect disentangling unitary $U$ exists and provided a proof for this result. Errors will occur because of the imperfect ansatz and the optimization process. When the optimization is complete, we can use $\rho_B$ and $\sigma_B$ for quantum entropy and distance measure estimation.

\begin{proposition}\label{prop1}
    For an $n$-qubit mixed state $\rho_{AB}$, if system $A$ is set to $(n - \log{r})$ qubits and system $B$ is set to $\log{r}$ qubits, where ${r} = 2^{\lceil \log \tilde{r} \rceil}$ with $\tilde{r} = \textnormal{rank}(\rho_{AB})$, then there exists a type-\RN{1} perfect disentangling unitary $U$ for this configuration. If $n \neq \log{r}$ (i.e., for quantum states that are not close to full rank), a perfect disentangling unitary still exists. (Note that all logarithms are to base 2, and since the number of qubits must be a natural number, we choose the value of $r$ to be a number that is the closest power of 2 to the rank.)
\end{proposition}
\begin{proof}
See the details of the proof in \textbf{Appendix~\ref{sec:proofprop1}}. 
\end{proof}

\begin{proposition}\label{prop2}
    For $n$-qubit mixed states $\rho_{AB}$ and $\sigma_{AB}$, if system $A$ is set to $(n - \log {r})$ qubits and system $B$ is set to $\log {r}$ qubits, where $V = \{v_i\}_{i=1}^{r_1}$ is the set of basis vectors obtained from the spectral decomposition of the $r_1$-rank quantum state $\rho_{AB}$, and $W = \{w_i\}_{i=1}^{r_2}$ is the set for the $r_2$-rank quantum state $\sigma_{AB}$, with $\tilde{r} = \dim(\textnormal{span}\{V \cup W\})$ and ${r} = 2^{\lceil \log \tilde{r} \rceil}$, then there exists a type-\RN{2} perfect disentangling unitary $U$ for this configuration. If $n \neq \log{r}$, a perfect disentangling unitary still exists.
\end{proposition}
\begin{proof}
See the details of the proof in \textbf{Appendix~\ref{sec:proofprop2}}.
\end{proof}

Optimizing $C(\theta)$ could face severe challenges. One example is barren plateaus~\cite{mcclean2018barren}. It is known that global cost functions cause barren plateaus~\cite{cerezo2021cost}. If the dimension of system $A$ is large, the swap operator $S_A$ in equation~(\ref{eq:11}) is a {large local} observable, and DEQNN becomes vulnerable to barren plateaus. If the cost function has the form $C(\theta) = \text{Tr}(OV(\theta)\rho V^\dagger(\theta))$, and if $O$ is expressed as a sum of $\ell$-local operators (where $\ell$ represents the number of qubits each operator acts on), then for sufficiently small $\ell$ and a total of $n$ qubits, the gradient of the cost function vanishes polynomially with respect to $n$ when the number of layers is $\mathcal{O}(\log n)$. So, if the number of qubits $n$ is large, choosing the dimension of the discard system $A$ as $n - \log r$ qubits could be problematic. To avoid barren plateaus the dimension of $A$ should be constant. We can sequentially discard a constant number of qubits. We can discard the qubits until the dimension of the preserved system is larger than $\log{r}$, as defined in \textbf{Proposition~\ref{prop2}}.

\subsection{Generalization}
Our DEQNN can be extended to estimate quantum information quantities for any $m$ quantum states. The definition of a type-$m$ generalized $\epsilon$-approximate disentangling unitary below can be seen as a natural extension of \textbf{Definition~\ref{def2}}. Additionally, the conditions under which a single unitary can act as a perfect disentangling unitary for $m$ quantum states are specified in \textbf{Proposition~\ref{prop3}}.

\begin{definition}\label{def3}
For $n$-qubit mixed states $\{\rho_i\}_{i=1}^m$, a unitary transformation $U$ that satisfies for all $i\in\{1,\ldots,m\}$,
\begin{align}
\text{Tr}\left(U\rho_i U^{\dagger} \ket{0}\bra{0}_A \otimes I_B\right) \geq 1-\epsilon,
\end{align}
is termed a {generalized $\epsilon$-approximate $m$-disentangling unitary} for $\{\rho_i\}_{i=1}^m$. (Note that when $m=2$, we refer to it as a type-\RN{2} disentangling unitary for convenience.) If $\epsilon = 0$, then $U$ is called perfect $m$-disentangling unitary for $\{\rho_i\}_{i=1}^m$.
\end{definition}

\begin{proposition}\label{prop3}
    For $n$-qubit mixed states $\{\rho_i\}_{i=1}^m$, if system $A$ is set to $(n - \log {r})$ qubits and system $B$ is set to $\log {r}$ qubits, where $V_i = \{v_j\}_{j=1}^{r_i}$ is the set of basis vectors obtained from the spectral decomposition of the $r_i$-rank quantum state $\rho_i$, with $\tilde{r} = \dim(\textnormal{span}\{\cup_{i=1}^m V_i\})$ and ${r} = 2^{\lceil \log \tilde{r} \rceil}$, then there exists a perfect $m$-disentangling unitary $U$ for this configuration. If $n \neq \log {r}$, a perfect disentangling unitary still exists.
\end{proposition}
\begin{proof}
The proof can be derived as an extension of \textbf{Proposition~\ref{prop2}}, following the same approach as in \textbf{Appendix~\ref{sec:proofprop2}}.
\end{proof}

In this setting, the generalized cost function of the DEQNN to be used for estimating the physical quantities of $m$ quantum states is set as follows:
\begin{equation}\label{eq:gen_cost}
    C_m(\theta) = 1-\frac{1}{m^2} \sum_{1\leq i,j\leq m}\T(S_A (\rho_i(\theta)\otimes\rho_j(\theta))),
\end{equation}
where $\rho_i(\theta) = U(\theta) \rho_i U^{\dagger}(\theta)$ and $S_A$ is the swap operator acting on system $A$. Solving the optimization problem $\min_\theta C_m(\theta)$ is ultimately equivalent to finding a disentangling unitary for the $m$ quantum states $\{\rho_i\}_{i=1}^m$. The preserved state $\T_A(\rho_i(\theta))$ can be used to estimate the quantum entropies of $\rho_i$, and $\T_A(\rho_i(\theta))$ and $\T_A(\rho_j(\theta))$ can be used to estimate the distance measures of $\rho_i$ and $\rho_j$. If we only want to estimate the quantum entropies, then we can change $C_m(\theta)$ to:
\begin{equation}\label{eq:gen_cost2}
    C_m(\theta) = 1-\frac{1}{m} \sum_{1\leq i\leq m}\T(S_A (\rho_i(\theta)\otimes\rho_i(\theta))).
\end{equation}



{Suppose that $\{\rho_i\}_{i=1}^m$ are quantum states of small rank that can be prepared using low-depth circuits. Our learning algorithm could then offer an advantage in situations where we want to simultaneously compute quantum entropy and distance measures for a large $m$ (i.e., for a large number of quantum states). Let the set of bases obtained from the spectral decomposition of the quantum state $\rho_i$ be ${V}_i=\{{v}_i\}$, and define ${S}_i=\text{span}\{{V}_i\}$ and ${\Omega}=\bigcup_{i=1}^m {S}_i$. If the condition $|{\Omega}| \le |{S}_i|,~\forall i \in \{1, \ldots, m\}$ (i.e., the condition where the cardinality does not increase even when taking the union of the span sets) is satisfied, then the quantum neural network $U(\theta)$ for learning the disentangling unitary of $\{\rho_i\}_{i=1}^m$ with the cost function $C(\theta)$ in equation~\eqref{eq:gen_cost2} will also have low complexity. In this case, estimating each quantity for every quantum state $\rho_i$ at once through a low-complexity quantum neural network is possible, which would be less costly than learning individual unitary operators for each quantum state. While this might generally be considered a rather strong condition, if satisfied, there should be no issue with trainability without considering the circuit complexity.}

\begin{figure}[t]
    \includegraphics[width=8.5cm]{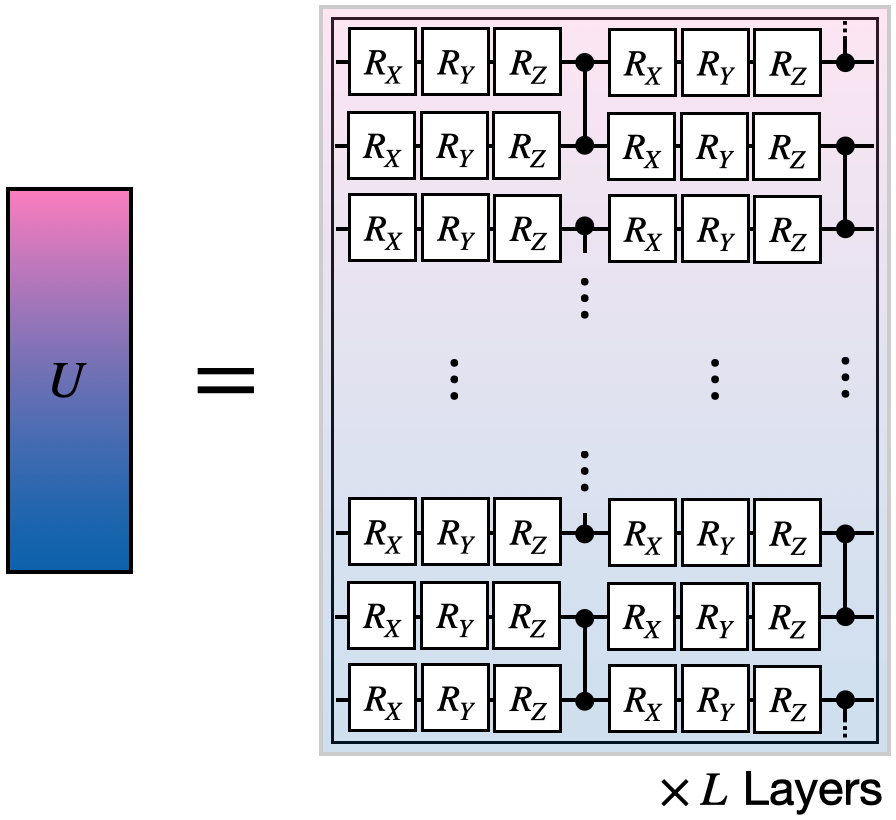}
    \caption{\textbf{Illustration of the ansatz used in the numerical simulation.} This unitary operator is applied to the $U$ part of the disentangling quantum neural network depicted in Figures 1 and 2, with measurements performed based on the desired system partitioning.}
    \label{fig:ansatz}
    \centering
\end{figure}

\section{Numerical Simulation}\label{sec:simulation}
We conducted a numerical simulation to demonstrate that the disentanglement process of DEQNN preserves entropies and distance measures. The simulation results are divided into two categories based on the number of qubits and the rank of the density matrix. The ansatz used in the numerical simulation is shown in Figure~\ref{fig:ansatz}, and the circuit structure can be modified as needed. To minimize issues related to barren plateaus, we used a quantum circuit with $\mathcal{O}(1)$ depth. For a more detailed explanation of the numerical simulation, refer to \textbf{Appendix~\ref{app:simulation}}.

\begin{figure*}
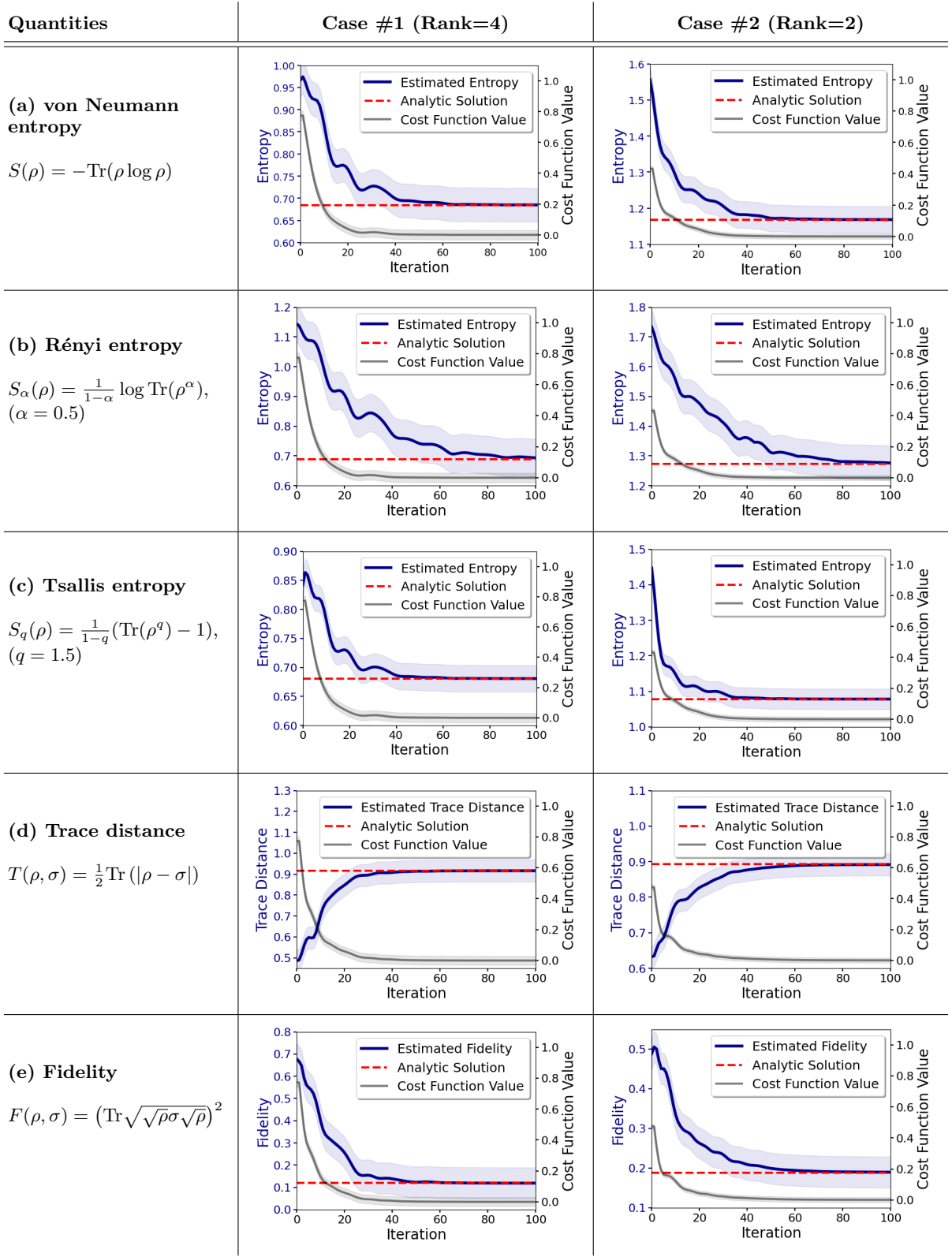

\renewcommand{\arraystretch}{1.8}
\centering
\resizebox{.9\linewidth}{!}{
\begin{tabular}{m{3.8cm}|c|c}
\textbf{Quantities} & \textbf{Case \#1 (Rank=4)} & \textbf{Case \#2 (Rank=2)} \\
\hline\hline
\raggedright \textbf{(a) von Neumann entropy \\~\\ $S(\rho)=-\T(\rho\log\rho)$}  & \fourvonone & \fourvontwo \\
\hline
\raggedright\textbf{(b) R\'enyi entropy \\~\\ $S_\alpha(\rho)=\frac{1}{1-\alpha}\log\T(\rho^\alpha),$ \\ $(\alpha=0.5)$}  & \fourrenyione & \fourrenyitwo \\
\hline
\raggedright\textbf{(c) Tsallis entropy \\~\\ $S_q(\rho)=\frac{1}{1-q}(\T(\rho^q)-1),$ \\ $(q=1.5)$} & \fourtsallisone & \fourtsallistwo \\
\hline
\raggedright\textbf{(d) Trace distance \\~\\ $T(\rho, \sigma) = \frac{1}{2}\T\left(|\rho-\sigma|\right)$ \\ ~}  & \fourtraceone & \fourtracetwo \\
\hline
\raggedright \textbf{(e) Fidelity \\~\\ $F(\rho,\sigma) = \left(\T\sqrt{\sqrt{\rho}\sigma\sqrt{\rho}}\right)^2$ \\ ~}  & \fourfidelityone & \fourfidelitytwo
\end{tabular}
}
\caption{\textbf{Four-qubit examples.} Convergence of the quantum neural estimation algorithm for the von Neumann entropy, R\'{e}nyi entropy ($\alpha=0.5$), Tsallis entropy ($q=1.5$), trace distance, and fidelity. Two input instances are randomly generated for each of these quantities. For each instance, the estimated values are presented in blue, the analytic solutions in red, and the cost function values of DEQNN in gray.}
\label{fig:fourqubit}
\end{figure*}

\begin{figure*}
\renewcommand{\arraystretch}{1.8}
\centering
\resizebox{.9\linewidth}{!}{
\begin{tabular}{m{3.8cm}|c|c}
\textbf{Quantities} & \textbf{Case \#1 (Rank=4)} & \textbf{Case \#2 (Rank=2)} \\
\hline\hline
\raggedright \textbf{(a) von Neumann entropy \\~\\ $S(\rho)=-\T(\rho\log\rho)$}  & \sixvonone & \sixvontwo \\
\hline
\raggedright\textbf{(b) R\'{e}nyi entropy \\~\\ $S_\alpha(\rho)=\frac{1}{1-\alpha}\log\T(\rho^\alpha),$ \\ $(\alpha=0.5)$}  & \sixrenyione & \sixrenyitwo \\
\hline
\raggedright\textbf{(c) Tsallis entropy \\~\\ $S_q(\rho)=\frac{1}{1-q}(\T(\rho^q)-1),$ \\ $(q=1.5)$} & \sixtsallisone & \sixtsallistwo \\
\hline
\raggedright\textbf{(d) Trace distance \\~\\ $T(\rho, \sigma) = \frac{1}{2}\T\left(|\rho-\sigma|\right)$ \\ ~}  & \sixtraceone & \sixtracetwo \\
\hline
\raggedright \textbf{(e) Fidelity \\~\\ $F(\rho,\sigma) = \left(\T\sqrt{\sqrt{\rho}\sigma\sqrt{\rho}}\right)^2$ \\ ~}  & \sixfidelityone & \sixfidelitytwo
\end{tabular}
}
\caption{\textbf{Six-qubit examples.} Convergence of the quantum neural estimation algorithm for the von Neumann entropy, R\'{e}nyi entropy ($\alpha=0.5$), Tsallis entropy ($q=1.5$), trace distance, and fidelity. Two input instances are randomly generated for each of these quantities. For each instance, the estimated values are presented in blue, the analytic solutions in red, and the cost function values of DEQNN in gray.}
\label{fig:sixqubit}
\end{figure*}

Figure~\ref{fig:fourqubit} shows the experiments for 4-qubit states, and Figure~\ref{fig:sixqubit}  shows the experiments for 6-qubit states. Each figure includes two cases: Case \#1 represents the scenario where the rank is 4, and Case \#2 represents the scenario where the rank is 2. These cases illustrate whether the convergence of DEQNN accurately estimates the ground-truth values and compare the results based on the number of qubits and the rank of the states.
\begin{itemize}
    \item \textbf{von Neumann entropy:} Figure~\ref{fig:fourqubit}-(a) and Figure~\ref{fig:sixqubit}-(a) present the results for the von Neumann entropy on 4-qubit and 6-qubit states. In Case \#1 and Case \#2, the given quantum states were compressed into a 2-qubit state and a 1-qubit state, respectively. Stable convergence occurred at around 100 epochs for 4-qubit states and 500 epochs for 6-qubit states.
    \item \textbf{R\'{e}nyi entropy:} Figure~\ref{fig:fourqubit}-(b) and Figure~\ref{fig:sixqubit}(b) present the results for the R\'{e}nyi entropy on 4-qubit and 6-qubit states, where $\alpha = 0.5$. The number of original and compressed qubits and the rank are the same as in the cases of von Neumann entropy. Stable convergence occurred at around 100 epochs for 4-qubit states and 500 epochs for 6-qubit states.
    \item \textbf{Tsallis entropy:} Figure~\ref{fig:fourqubit}-(c) and Figure~\ref{fig:sixqubit}-(c) present the results for Tsallis entropy on 4-qubit and 6-qubit states, where $q = 1.5$. The number of original and compressed qubits and the rank are the same as in the cases of von Neumann and R\'{e}nyi entropies. Stable convergence occurred at around 100 epochs for 4-qubit states and 500 epochs for 6-qubit states.
    \item \textbf{Trace distance:} Figure~\ref{fig:fourqubit}-(d) and Figure~\ref{fig:sixqubit}-(d) present the results for trace distance on 4-qubit and 6-qubit states. We compressed the given 4-qubit states into 2-qubit states in Case \#1, and 6-qubit states into 3-qubit states in Case \#2. Stable convergence occurred at around 100 epochs for 4-qubit states and 800 epochs for 6-qubit states.
    \item \textbf{Fidelity:} Figure~\ref{fig:fourqubit}-(e) and Figure~\ref{fig:sixqubit}-(e) present the results for fidelity on 4-qubit and 6-qubit states. The conditions regarding the number of qubits and the rank are the same as those for trace distance. Stable convergence occurred at around 100 epochs for 4-qubit states and 800 epochs for 6-qubit states.
\end{itemize}

\section{Discussion} \label{sec:discussion}
In this work, we proposed disentangling quantum neural networks (DEQNN), a quantum neural network structure for dimension reduction that can be used for efficient estimation of quantum entropies and distance measures. The continuity bounds and disentanglement serve crucial roles in designing the cost function of DEQNN. We prove that optimizing the cost function of DEQNN leads to the preservation of quantum entropies and distance measures in a smaller partial state. We can use the partially preserved state for estimation. {This scheme is useful for the common cases when we have to iteratively estimate the generalized entropy to figure our the proper parameter of the system.} Furthermore, our method can be generalized to estimate quantities for an arbitrary number of states. We propose that DEQNN is particularly beneficial when estimating quantities for a large number of less complex quantum states. Our results are supported by numerical simulations. However, since our algorithm does not propose a generalized statistical mechanics, caution is warranted regarding its applicability—even if stable results are observed in simulations. As discussed in Section~\ref{sec:background}, attention must be given to the lack of a thermodynamic limit and Lesche stability.

However, further investigation is needed into the trainability and expressivity of our method. Exploring applications of DEQNN beyond von Neumann entropy, Rényi entropy, Tsallis entropy, trace distance, and fidelity, as well as expanding our proofs to other measures, could be an interesting avenue for future research. Specifically, examining how our approach could be applied to entropy estimations with distinct mathematical structures, such as Havrda-Charvát entropy, is an important direction. For instance, while both Havrda-Charvát entropy and Rényi entropy can be viewed as specific forms of the Nagumo-Kolmogorov functional, they differ in mathematical structure—particularly in terms of differentiability and continuity at boundary limits. As a result, attempts to estimate both quantities with a unified algorithm may lead to inaccuracies or instability. Although our study does not address Havrda-Charvát entropy, it is necessary to investigate dimensionality reduction methods that encompass this type of entropy, as well as the potential existence of a unified algorithm for physical quantity estimation (or the fundamental reasons for its non-existence).

{Another important topic for future research is related to the complexity of the quantum neural network, ensuring convergence for the same error bound $\epsilon$ as the number of qubits and rank changes. If the number of layers required to achieve a specific error bound can be shown to be $\mathcal{O}(\text{poly}(\log n))$ or less using a constant-depth ansatz, our learning algorithm may avoid barren plateaus. Proving this rigorously would provide key insights into the trainability of the network, though this remains a challenging and highly valuable topic for future research. One possible approach is to leverage recent studies that construct approximate unitary designs using shallow circuits~\cite{schuster2024random}. Following the methodology proposed in that study, we could consider constructing the ansatz by attaching local random unitaries to $\mathcal{O}(\log n)$-sized or $\mathcal{O}(\text{poly}(\log n))$-sized patches of qubits, thereby creating a global random unitary. If this unitary can describe an $\epsilon$-approximate disentangling unitary, then our quantum neural network could be implemented with low complexity.}

\section*{Data Availability Statement}
{The data and software that support the findings of this study
can be found in the following repository: \texttt{https://github.com/tfoseel/DEQNN}.}

\section*{Acknowledgments}
We acknowledge Daniel Kyungdeock Park and Ju-Young Ryu for their valuable discussions. This work was supported by the National Research Foundation of Korea (NRF) through a grant funded by the Ministry of Science and ICT (NRF-2022M3H3A1098237). This work was partially supported by the Institute for Information \& Communications Technology Promotion (IITP) grant funded by the Korean government (MSIP) (No. 2019-0-00003; Research and Development of Core Technologies for Programming, Running, Implementing, and Validating of Fault-Tolerant Quantum Computing Systems), and Korea Institute of Science and Technology Information (KISTI: P24021).

\section*{Author Contributions}
M.S., S.L., and J.L. contributed equally to this work, undertaking the primary responsibilities, including the development of the main ideas, mathematical proofs, initial drafting, and revisions of the paper. M.L. and D.J. contributed to the numerical simulations and paper preparation. H.Y. conducted the analysis from a physical perspective, including the physical limitations of the algorithm, and participated in the revision of the paper. K.J. supervised the research. All authors discussed the results and contributed to the final paper.

\appendix

\section{Proofs of Continuity Bounds}\label{sec:proofcontbound}
We complete the proofs of continuity bounds in Table~\ref{table:cont_bound}. Most of the continuity bounds are proved in the references. We tighten the bounds of R\'{e}nyi and Tsallis entropy in terms of rank and prove the bounds of fidelity using the triangle inequality.

\subsection{von Neumann entropy}
The continuity bounds of von Neumann entropy are proved by Audenaert~\cite{audenaert2007sharp}. Audenaert proved that for $d$-dimensional probability distributions $p=\{p_i\}$ and $q=\{q_i\}$, following holds:
\begin{align}
& |H(p)-H(q)|  \nonumber \\
& \leq D\log(d-1)-D\log D-(1-D)\log (1-D),
\end{align}
where $D=\frac{1}{2}\sum^d_{i=1}|p_i-q_i|$, $H(p) = \sum_i p_i\log p_i$, $H(q) = \sum_i q_i\log q_i$. We can easily check that replacing dimension $d$ with rank $r$, which is the number of non-zero values of the distribution, does not change $D, H(p)$ or $H(q)$. Therefore, it still satisfies Audenaert's inequality. For quantum states $\rho$ and $\sigma$, whose eigenvalues are $\{p_i\}$, $\{q_i\}$, $D \leq T(\rho, \sigma) = T$, $H(p) = S(\rho)$ and $H(q) = S(\sigma)$ hold. We can deduce the equation below:
\begin{align}
& |T_{\alpha}(\rho)-T_{\alpha}(\sigma)| \nonumber \\
& \leq T\log(r-1)-T\log T-(1-T)\log (1-T).
\end{align}

\subsection{R\'{e}nyi and Tsallis entropies}
The continuity bounds of R\'{e}nyi case can follow Tsallis entropy. Therefore, we first focus on the proof of the continuity bound on Tsallis entropy. We modified the proof by Chen \emph{et al.}~\cite{chen2017sharp}. For clarification, we denote Tsallis entropy as $T_{\alpha}$ and R\'{e}nyi entropy as $S_{\alpha}$ only for this section. Chen proved that for $\alpha < 1$, $d$-dimensional probability distributions $p=\{p_i\}$ and $q=\{q_i\}$, the following holds:
\begin{equation}
|H_{\alpha}(p)-H_{\alpha}(q)| \leq \frac{(1-D)^{\alpha}-1+(d-1)^{1-\alpha}D^{\alpha}}{1-\alpha},
\end{equation}
where $D=\frac{1}{2}\sum^d_{i=1}|p_i-q_i|$, $H_{\alpha}(p) = \frac{1}{1-\alpha}\sum_i (p_i)^{\alpha}-1$ and $H_{\alpha}(q) = \frac{1}{1-\alpha}\sum_i (q_i)^{\alpha}-1$. We  can easily see that replacing dimension $d$ with rank $r$, which is the number of non-zero values of the distribution, does not change $D$, $H_{\alpha}(p)$, or $H_{\alpha}(q)$. Therefore, it still satisfies Chen's inequality. 

For quantum states $\rho$ and $\sigma$, whose eigenvalues are $\{p_i\}$ and $\{q_i\}$, respectively, $D \leq T(\rho, \sigma) = T$, $H_{\alpha}(p) = T_{\alpha}(\rho)$, and $H_{\alpha}(q) = T_{\alpha}(\sigma)$ hold. We can deduce the equation as follow:
\begin{equation}
|T_{\alpha}(\rho)-T_{\alpha}(\sigma)| \leq \frac{(1-T)^{\alpha}-1+(r-1)^{1-\alpha}T^{\alpha}}{1-\alpha}.
\end{equation}

For $\alpha > 1$, Raggio~\cite{raggio1995properties} proved the continuity bound for Tsallis entropy. Thus, the continuity bound for Tsallis entropy in Table~\ref{table:cont_bound} is proved.

The difference in R\'{e}nyi entropies can be bounded by Tsallis entropy using the mean value theorem:
\begin{equation}
|S_{\alpha}(\rho)-S_{\alpha}(\sigma)| \leq \frac{|T_{\alpha}(\rho)-T_{\alpha}(\sigma)|}{\min\{\T(\rho^{\alpha}),\T(\sigma^{\alpha})\}}.
\end{equation}

Since $\T(\rho^{\alpha}),\T(\sigma^{\alpha}) \geq r^{1-\alpha}$, the continuity bound for R\'{e}nyi entropy in Table~\ref{table:cont_bound} is proved.

\subsection{Trace Distance and Fidelity}
The continuity bound of trace distance is established using the triangle inequality for trace distances. Below, we elaborate on the proof of the continuity bound of fidelity. For any $x,y \in [0, 1]$, $|x-y| \leq |\arccos x - \arccos y|$ is obtained by the mean value theorem. Combining this with the triangle inequality for fidelity angles in equation~(\ref{eq:2}), we get:
\begin{equation}
|F(\rho_1,\sigma_1)-F(\rho_2,\sigma_2)| \leq A(\rho_1,\rho_2) + A(\sigma_1,\sigma_2),
\end{equation}
where $A = \arccos F$. Using the inequality $F(\rho, \sigma) \leq 1-T(\rho, \sigma)^2$, we can complete the proof.


\section{Proof of Lemma 1}\label{sec:prooflem1}
\begin{proof}
Let $\rho_{AB}^* = U\rho_{AB}U^{\dagger}$, $\rho_A^* = \T_B(\rho_{AB}^*)$ and $\rho_B^* = \T_A(\rho_{AB}^*)$. Using the definition of Hilbert-Schmidt trace distance, $D = D_{HS}(\rho_{AB}^*, \ket{0}_A\bra{0}_A\otimes\rho_B^*)$ can be represented as:
\begin{equation}
\T({\rho_{AB}^*}^2)+\T(\rho^2_B)-2\T(\rho_{AB}^*\ket{0}_A\bra{0}_A\otimes\rho_B^*).
\end{equation}
$\T(\rho_{AB}^*\ket{0}_A\bra{0}_A\otimes\rho_B^*)$ can be rearranged as: 
\begin{equation}
\T(\rho_{AB}^* I_A\otimes\rho_B^*)-\T(\rho_{AB}^* (I_A-\ket{0}\bra{0}_A)\otimes\rho_B^*).
\end{equation}
Next, by exploiting
\begin{equation}
\T(\rho_{AB}^* I_A\otimes\rho_B^*) = \T({\rho_B^*}^2),
\end{equation}
and 
\begin{align}
& \T(\rho_{AB}^* (I_A-\ket{0}\bra{0}_A)\otimes\rho_B^*) \\
& ~~~\leq \T(\rho_{AB}^* (I_A-\ket{0}\bra{0}_A)\otimes I_B) \\
& ~~~ = 1-\T(\rho_{AB}^* -\ket{0}\bra{0}_A\otimes I_B) \leq \epsilon,
\end{align}
we have $D \leq \T({\rho_{AB}^*}^2)-\T({\rho_B^*}^2)+2\epsilon$. The subadditivity of Tsallis entropy~\cite{audenaert2007subadditivity} implies that
\begin{equation}
\T({\rho_{AB}^*}^2)-\T({\rho_{B}^*}^2) \leq 1-\T({\rho_{A}^*}^2).
\end{equation}
Let $\rho_A^* = \sum_i p_i\ket{\psi_i}\bra{\psi_i}$ and $p_{\text{max}} = \max_i p_i$. Since $\bra{0}\rho_A\ket{0} = \T(\rho_{AB}^*\ket{0}_A\bra{0}_A\otimes I_B) \geq 1-\epsilon$, we have
\begin{equation}
p_{\text{max}} \geq \sum_i p_i \inn{0}{\psi_i}^2 \geq 1-\epsilon.
\end{equation}
Then,
\begin{equation}
\T({\rho_A^*}^2) = \sum_i p^2_i \geq p_{\text{max}}^2 \geq (1-\epsilon)^2.
\end{equation}
So, we can finally deduce that:
\begin{equation}
D \leq 2\epsilon-\epsilon^2+2\epsilon \leq 4\epsilon.
\end{equation}
Now, we use the relation between Hilbert-Schmidt trace distance and trace distance~\cite{coles2019strong}:
\begin{equation}
T(\rho, \sigma)^2 \leq \frac{\text{rank}(\rho)\text{rank}(\sigma)D_{HS}(\rho, \sigma)}{\text{rank}(\rho)+\text{rank}(\sigma)} \leq \text{rank}(\rho)D.\\
\end{equation}
Therefore,
\begin{equation}
T(\rho_{AB}^*, \ket{0}_A\bra{0}_A\otimes\rho_B^*) \leq 2\sqrt{r\epsilon}.
\end{equation}
\end{proof}

\section{Proofs of Theorems 1 and 2}\label{sec:proofthms}
We prove the inequalities on Table~\ref{table:ineq}. We apply \textbf{Lemma~\ref{lemma:td}} and then rearrange the continuity bounds in Table~\ref{table:cont_bound}. As mentioned before, the purpose of \textbf{Theorems~\ref{thm:1} and \ref{thm:2}} is to show that disentanglement preserves quantum entropies and distance measures with polynomial error. Therefore, the inequalities in Table~\ref{table:ineq} are less stringent than the continuity bounds for better understanding.
\begin{proof}
For $f \in F_1=\{S,S_{\alpha},S_q\}$, it is known that
\begin{align}
f(U\rho_{AB}U^{\dagger}) &= f(\rho_{AB}), \\
f(\ket{0}\bra{0}_A\otimes\rho_B^*) &= f(\rho_B^*).
\end{align}
Thus,
\begin{equation}
|f(\rho_{AB})-f(\rho_B^*)| = |f(U\rho_{AB}U^{\dagger})-f(\ket{0}\bra{0}_A\otimes\rho_B^*)|.
\end{equation}
Hence, we can apply the continuity bound and set $T_{\rho} = 2\sqrt{r\epsilon}$ in Table~\ref{table:cont_bound}. For similar reasons, for $f \in F_1=\{S,S_{\alpha},S_q\}$
\begin{align}
& |f(\rho_{AB},\sigma_{AB})-f(\rho_B^*,\sigma_B^*)| \nonumber \\
& = |f(U\rho_{AB}U^{\dagger}, U\sigma_{AB}U^{\dagger})-f(\ket{0}\bra{0}_A\otimes\rho_B^*, \ket{0}\bra{0}_A\otimes\sigma_B^*)|
\end{align}
holds. Hence, we can apply the continuity bound and set $T_{\rho} = 2\sqrt{r\epsilon}$, $T_{\sigma} = 2\sqrt{r\epsilon}$ in Table~\ref{table:cont_bound}.

Now we need to rearrange the continuity bounds into a simpler form in \textbf{Theorems~\ref{thm:1} and \ref{thm:2}}. The rearrangement of R\'{e}nyi, Tsallis entropy, and trace distance is straightforward, so we elaborate on von Neumann entropy and fidelity.

Using elementary algebra, $\log(r-1) \leq \sqrt{r}$, $x\log\frac{1}{x} \leq \sqrt{x}$ and $(1-x)\log(1-x) \leq x$ $(x \ll 1)$ can be easily shown. Therefore, we can deduce the results in Table~\ref{table:ineq} by applying $T_{\rho} = \sqrt{\epsilon}$ to the continuity bounds of von Neumann entropy in Table~\ref{table:cont_bound}. By using the Taylor series of $\arccos$, we get $\arccos(1-x) \leq \frac{\pi}{2}x$. Using this and applying $T_{\rho} = 2\sqrt{r\epsilon}$ and $T_{\sigma} = 2\sqrt{r\epsilon}$, the inequality for fidelity in Table~\ref{table:ineq} is concluded.
\end{proof}

\section{Proof of Lemma 2}\label{sec:prooflem2}
\begin{proof}
For convenience, let $\rho = \rho_A(\theta) = \sum_i p_i\ket{\psi_i}\bra{\psi_i}$ and $\sigma = \sigma_A(\theta) = \sum_i q_i\ket{\phi_i}\bra{\phi_i}$, where $p_1 \geq p_2 \geq \cdots$ and $q_1 \geq q_2 \geq \cdots$, respectively. Also, assume that $C(\theta) < \epsilon$. 

Then, $\T(\rho^2)+\T(\sigma^2)+2\T(\rho\sigma) > 4-4\epsilon$. Hence, $\T(\rho^2) > 1-4\epsilon$, $\T(\sigma^2) > 1-4\epsilon$ and $\T(\rho\sigma) > 1-2\epsilon$. 

We can observe that
\begin{equation}
p^2_1+(1-p_1)^2 = \sum_i p^2_i = \T(\rho^2) > 1-4\epsilon, 
\end{equation}
then $p_1 > \frac{1+\sqrt{1-8\epsilon}}{2} > 1-4\epsilon$. Hence, $\ket{\psi_1}$ satisfies
\begin{equation}
\bra{\psi_1}\rho\ket{\psi_1} > 1-4\epsilon,
\end{equation}
and with the same logic, $\ket{\phi_1}$ satisfies
\begin{equation}
\bra{\phi_1}\sigma\ket{\phi_1} > 1-4\epsilon.
\end{equation}

Now, we investigate the relation between $\ket{\psi_1}$ and $\ket{\phi_1}$ using $\T(\rho\sigma) > 1-2\epsilon$. That is,
\begin{equation}
\sum_i p_i q_i \inn{\psi_i}{\phi_i}^2 = \T(\rho\sigma) > 1-2\epsilon
\end{equation}
holds. Since $\sum_{i\geq 2} p_i < 4\epsilon$ and $\sum_{i\geq 2} q_i < 4\epsilon$,
\begin{equation}
\inn{\psi_1}{\phi_1}^2 + 8\epsilon > \sum_i p_i q_i \inn{\psi_i}{\phi_i}^2 > 1-2\epsilon. 
\end{equation}
Thus, we can finally deduce that
\begin{equation}
\bra{\psi_1}\sigma\ket{\psi_1} \geq q_1\inn{\psi_1}{\phi_1}^2 > 1-14\epsilon,
\end{equation}
where $\ket{\psi_1}$ corresponds to $\ket{\psi}$ in \textbf{Lemma~\ref{lemma:cost}}, and we conclude the proof.
\end{proof}

\section{Proof of Proposition 1}\label{sec:proofprop1}
\begin{proof}    
    Let the quantum state $\rho_{AB}$ be expressed as $\rho_{AB} = \sum_{i=1}^r p_i |\psi_i\rangle\langle \psi_i|$, where $\{|\psi_i\rangle\}_{i=1}^r$ are orthogonal states associated with probabilities $\{p_i\}$. There exist $r$ orthogonal states acting on $\log r$ qubits, denoted by $\{|\phi_i\rangle\}_{i=1}^r$. For each $i$, define the following extended states:
    \begin{align}
        |\phi_i'\rangle = |\phi_i\rangle |0\rangle^{\otimes (n - \log r)},
    \end{align}
    so that $\{|\phi_i'\rangle\}_{i=1}^r$ are orthogonal states acting on the full $n$-qubit Hilbert space. Next, extend $\{|\psi_i\rangle\}_{i=1}^r$ and $\{|\phi_i'\rangle\}_{i=1}^r$ to complete orthonormal bases $\{|\psi_i\rangle\}_{i=1}^{2^n}$ and $\{|\phi_i'\rangle\}_{i=1}^{2^n}$ for the $n$-qubit Hilbert space by applying the Gram-Schmidt orthogonalization process.\\
    Now define $U = \sum_{i=1}^{2^n} |\phi_i'\rangle \langle \psi_i|$. This operator $U$ satisfies $UU^\dagger = U^\dagger U = I$, confirming that $U$ is unitary. Finally, applying $U$ to $\rho_{AB}$ yields:
    \begin{align}
        U \rho_{AB} U^\dagger &= \sum_{i=1}^r p_i |\phi_i'\rangle \langle \phi_i'| \\
        &=  \left( \sum_{i=1}^r p_i |\phi_i\rangle \langle \phi_i| \right) \otimes \left( |0\rangle \langle 0| \right)^{\otimes (n - \log r)}\\
        &= \rho_B^* \otimes \left( |0\rangle \langle 0|_{A} \right)^{\otimes (n - \log r)},
    \end{align}
which demonstrates perfect disentanglement, where system $A$ is separated as $(n-\log r)$-qubit subsystem and system $B$ as a $\log r$-qubit subsystem.
\end{proof}

\section{Proof of Proposition 2}\label{sec:proofprop2}
\begin{proof}    
    Let two quantum states $\rho_{AB}$ and $\sigma_{AB}$ be expressed as
    \begin{align}
        \rho_{AB} = \sum_{i=1}^{r_1} p_i |v_i\rangle\langle v_i|,~\text{and}~\sigma_{AB}=\sum_{i=1}^{r_2}q_i |w_i\rangle\langle w_i|.    
    \end{align}
    Define ${S}=\text{span}\{|v_1\rangle, \ldots, |v_{r_1}\rangle, |w_1\rangle, \ldots, |w_{r_2}\rangle\}$, and let $\{\psi_i\}_{i=1}^k$ be a basis for ${S}$, where $k \le r$.\\
    There exist $r$ orthogonal states acting on $\log r$ qubits, denoted by $\{|\phi_i\rangle\}_{i=1}^r$. For each $i$, define the following extended states:
    \begin{align}
        |\phi_i'\rangle = |\phi_i\rangle |0\rangle^{\otimes (n - \log r)},
    \end{align}
    so that $\{|\phi_i'\rangle\}_{i=1}^r$ are orthogonal states spanning the full $n$-qubit Hilbert space. Next, extend $\{|\psi_i\rangle\}_{i=1}^k$ and $\{|\phi_i'\rangle\}_{i=1}^r$ to complete orthonormal bases $\{|\psi_i\rangle\}_{i=1}^{2^n}$ and $\{|\phi_i'\rangle\}_{i=1}^{2^n}$ for the $n$-qubit Hilbert space by applying the Gram-Schmidt orthogonalization process.\\
    Now define $U = \sum_{i=1}^{2^n} |\phi_i'\rangle \langle \psi_i|$. This operator $U$ satisfies $UU^\dagger = U^\dagger U = I$, confirming that $U$ is unitary.\\
    Since $U|\psi_i\rangle=|\phi_i'\rangle=|\phi_i\rangle|0\rangle^{\otimes{(n-\log r)}}$, there exists a state $|v_i'\rangle$ such that:
    \begin{align}
        U|v_i\rangle &= U \sum_{j=1}^k c_{ij}|\psi_j\rangle\\
        &= \left(\sum_{j=1}^k c_{ij} |\phi_j\rangle\right) |0\rangle^{\otimes{(n-\log r)}}\\
        &= |v_i'\rangle|0\rangle^{\otimes{(n-\log r)}}.
    \end{align}
    Similarly, there exists a state $|w_i'\rangle$ such that $U|w_i\rangle=|w_i'\rangle|0\rangle^{\otimes{(n-\log r)}}$. Finally, applying $U$ to $\rho_{AB}$ and $\sigma_{AB}$ yields:
    \begin{align}
        U\rho_{AB} U^\dagger &= \left(\sum_{i=1}^{r_1} p_i |v_i'\rangle\langle v_i'|\right)\otimes\left(|0\rangle\langle 0|\right)^{\otimes n-\log r} \\
        &= \rho_B^* \otimes \left( |0\rangle \langle 0|_{A} \right)^{\otimes (n - \log r)}~\text{and}, \\
        U\sigma_{AB} U^\dagger &= \left(\sum_{i=1}^{r_2} q_i |w_i'\rangle\langle w_i'|\right)\otimes\left(|0\rangle\langle 0|\right)^{\otimes n-\log r} \\
        &= \sigma_B^* \otimes \left( |0\rangle \langle 0|_{A} \right)^{\otimes (n - \log r)},
    \end{align}
    which demonstrates perfect disentanglement, where system $A$ is separated as $(n-\log r)$-qubit subsystem and system $B$ as a $\log r$-qubit subsystem.
\end{proof}

\section{Further Details on Numerical Simulations}\label{app:simulation}
{In this section, we provide a more detailed explanation of the process for finding our $\epsilon$-approximate disentangling unitary $U$, including the steps involved in the numerical simulation. As described in the main text, our goal is to identify $U$ that enables dimensionality reduction while preserving the quantum information quantities. We propose achieving this using a variational quantum algorithm approach implemented through a quantum neural network (or parametrized quantum circuit). Here, the variational quantum algorithm is also referred to as a learning algorithm, as we use these terms interchangeably to describe the process of optimizing the parameters of the rotation gates in the quantum neural network through gradient descent to solve the optimization problem for the objective function and approximate the desired $U$. Our learning algorithm is summarized in the pseudocode of \textbf{Algorithm 1}.}

\begin{algorithm}[H]
    \caption{Variational algorithm for finding (learning) $\epsilon$-approximate disentangling unitary $U$}
    {\textbf{Input:}
    \begin{itemize}
        \item $\mathcal{O}(1)$-depth parametrized quantum circuit with randomly initialized parameter vector $\theta^0$
        \item Copies of $m$ quantum states $\{\rho_i\}_{i=1}^m$
        \item Number of iterations $n_\text{step}$
        \item Learning rate $\alpha$
    \end{itemize}
    \textbf{Repeat the following process for $n_\text{step}$ iterations:}
    \begin{enumerate}
    \item Evaluate $\nabla C_m(\theta^i)$, as defined in equation~\eqref{eq:gen_cost}, using parameter-shift rule at the $i$-th step.
    \item Update parameter vector $\theta^{i+1} = \theta^i - \alpha \nabla_{\theta} C(\theta^i)$.
    \end{enumerate}
    \textbf{Output:} $\epsilon$-approximate disentangling unitary $U$ composed of the optimized parameter vector $\theta^{n_\text{step}}$}
\end{algorithm}

{The detailed explanation of each part of the algorithm is as follows. Additionally, you can find further details related to the implementation through the GitHub repository specified in the Data Availability section.
\begin{itemize}
    \item \textbf{Parameter vector and initialization:} The parameter vector at the $i$-th step can be represented as an $aL$-dimensional vector,
        \begin{align}
            {\theta}^i &= (\theta^i_{(j,k)})_{1 \le j \le L, 1 \le k \le a} \\
            &= (\theta^i_{(1,1)}, \ldots, \theta^i_{(1,a)}, \ldots, \theta^i_{(L,1)}, \ldots, \theta^i_{(L,a)})
        \end{align}
    where $a$ denotes the number of parameters in a single ansatz, and $L$ is the number of layers. For initialization, $\theta^0_{(j, k)}$ is initially sampled from the interval $[0, 2\pi)$. The design of the ansatz we used can be found in Figure~\ref{fig:ansatz}. For an $n$-qubit system, a single layer contains $6n$ rotation gates, so when using $L$ layers, the dimension of the parameter vector is $6nL$ in our case. The ansatz design can be constructed using (trainable) arbitrary rotation gates and (constant) 2-qubit gates, and the dimensionality of the parameters depends on the specific form of the design.
    \item \textbf{Preparation of input quantum states:} To benchmark our algorithm, we need to generate quantum states for estimating quantum information quantities. To achieve this, we sample a unitary operator $U_\text{init}$ from the unitary Haar measure. The sampled unitary operator is applied to a quantum state initialized to zero. Subsequently, we trace out part of the system to the desired extent, obtaining a mixed state. These randomly generated reduced density matrices are then used to validate the performance of our algorithm. (This process can be implemented using the \texttt{density\_matrix} and \texttt{QubitUnitary} functions from PennyLane~\cite{bergholm2018pennylane}, together with the \texttt{random\_unitary} function from the \texttt{quantum\_info} module of Qiskit.)
    \item \textbf{Gradient evaluation with the parameter-shift rule:} The problem we aim to solve is, strictly speaking, to find the theoretical optimal parameter vector $\theta_\text{exact} = \arg\min C_m(\theta)$. However, since solving this optimization problem exactly is infeasible, our goal is to perform interactive optimization of the parameter vector to find a solution close to the theoretical optimal. Specifically, as demonstrated by \textbf{Lemma~\ref{lemma:cost}} and \textbf{Theorem~\ref{thm:3}} in the main text, if we can make $C_m(\theta)$ smaller than $\epsilon$, the quantum information quantities are theoretically guaranteed to be well-preserved. Therefore, we design an iterative optimization algorithm that minimizes $C_m(\theta)$ through sufficient iterations, enabling the learning of an $\epsilon$-approximate disentangling unitary operator $U$ from the resulting parameter vector.\\
    The analytic gradient of quantum circuits with respect to gate parameters can be computed using the \textit{parameter-shift rule}~\cite{mitarai2018quantum, schuld2019evaluating}. This rule allows for exact gradient calculation by computing the difference between the results of shifting a single gate parameter twice, without altering the structure of the original circuit. Particularly, since the designed $C_m(\theta)$ can be easily computed via the swap test, the calculations required for gradient descent can be implemented very efficiently on quantum hardware.\\
    To explain in more technical detail, consider a parametrized quantum circuit $U(\theta)$, where the circuit can be interpreted as a sequence of $N$ unitary operators and represented as $U(\theta) = \prod_{j=1}^N U(\theta_j)$. For a function that evaluates the expectation value of an observable $\mathcal{M}$, given by $f(\theta) := \langle \mathcal{M} \rangle = \langle 0 | U^\dagger(\theta)\mathcal{M}U(\theta)|0\rangle$, it is known that if the Hermitian generator $\mathcal{G}$ of the unitary operator $U(\theta) = e^{-i\theta\mathcal{G}}$ has at most two eigenvalues $\pm\lambda$, the gradient obtained via the parameter-shift rule is $\nabla_{\theta_j} f(\theta_j) = \lambda \left(f\left(\theta_j + \frac{\pi}{4\lambda}\right) - f\left(\theta_j - \frac{\pi}{4\lambda}\right)\right)$~\cite{schuld2019evaluating}, which can be efficiently computed with just two measurements. If multiple parameters are present, the product rule can be used to calculate the final gradient $\nabla_\theta f(\theta)$. In our ansatz, which consists solely of Pauli rotation gates, $\lambda = \frac{1}{2}$ is determined, simplifying the calculation.\\
    Based on the analytic gradient obtained through this process, the parameters are updated using an initially specified learning rate $\alpha$. By iterating this process for a predetermined number of iterations $n_\text{step}$, a solution close to the theoretical optimal can be achieved. In some cases, if the value of $C_m(\theta)$ falls below a predefined threshold before reaching the maximum number of iterations and convergence is sufficient, the learning process can be terminated early. However, determining the optimal form of the ansatz for more efficient learning remains a non-trivial task.
\end{itemize}}

\bibliographystyle{ieeetr}
\bibliography{citation.bib}

\end{document}